\documentclass[letterpaper,11pt]{article}
\usepackage[utf8]{inputenc}
\usepackage[margin=1.0in]{geometry}
\usepackage{microtype}
\usepackage{amsmath,amsthm,amssymb, color, tikz, mathtools}
\usepackage{tcolorbox}
\usepackage{mdframed}

\usetikzlibrary{positioning}
\usetikzlibrary{shapes}
\usetikzlibrary{decorations.pathreplacing}
\usepackage{enumerate, enumitem} 

\usepackage{thm-restate}
\usepackage[hidelinks]{hyperref}
\usepackage{cleveref}

\title{On the Composition of Randomized Query Complexity and Approximate Degree}

\newcommand{\EE}{\mathbb{E}}

\newcommand{\set}[1]{\left\{ #1 \right\}}

\usepackage{bbm}

\newcommand{\zone}{\set{0,1}}
\newcommand{\zonep}{\set{0,1,*}}

\newcommand{\OR}{\mathsf{OR}}
\newcommand{\PrOR}{\mathsf{PrOR}}
\newcommand{\AND}{\mathsf{AND}}

\newcommand{\rub}{\mathsf{RUB}}
\newcommand{\sink}{\mathsf{SINK}}

\renewcommand{\tilde}{\widetilde}
\renewcommand{\Tilde}{\widetilde}
\newcommand{\bdeg}{\widetilde{\mathrm{bdeg}}}
\newcommand{\adeg}{\widetilde{\deg}}
\newcommand{\rqc}{\textnormal{R}}
\newcommand{\s}{\textnormal{s}}
\newcommand{\bs}{\textnormal{bs}}
\newcommand{\fbs}{\textnormal{fbs}}

\newcommand{\sign}{\textnormal{sign}}

\newtheorem{theorem}{Theorem}[section]
\newtheorem{corollary}[theorem]{Corollary}
\newtheorem{lemma}[theorem]{Lemma}
\newtheorem{defi}[theorem]{Definition}
\newtheorem*{question}{Question}
\newtheorem{observation}[theorem]{Observation}

\newtheorem{prop}[theorem]{Proposition}
\newtheorem{open question}[theorem]{Open question}

\newcommand{\noisyR}{\textnormal{noisyR}}
\newcommand{\Dom}{\textnormal{Dom}}
\newcommand{\R}{\textnormal{R}}
\newcommand{\GapMaj}{\textnormal{GapMaj}}

\author{
Sourav Chakraborty\thanks{Indian Statistical Institute, Kolkata, India. \texttt{sourav@isical.ac.in}}
\and 
Chandrima Kayal\thanks{Indian Statistical Institute, Kolkata, India. \texttt{chandrimakayal2012@gmail.com}}
\and
Rajat Mittal\thanks{Indian Institute of Technology Kanpur, India.  \texttt{rmittal@cse.iitk.ac.in }}
\and
Manaswi Paraashar\thanks{Aarhus University, Denmark. \texttt{manaswi.isi@gmail.com}}
\and 
Swagato Sanyal\thanks{Indian Institute of Technology  Kharagpur, India. \texttt{sanyalswagato@gmail.com}}
\and
Nitin Saurabh \thanks{Indian Institute of Technology, Hyderabad, India. \texttt{3295.nitin@gmail.com}}
}

\begin{document}

\maketitle

\begin{abstract}
For any Boolean functions $f$ and $g$, the question whether $\text{R}(f\circ g) = \tilde{\Theta}(\text{R}(f) \cdot \text{R}(g))$, is known as the composition question for the randomized query complexity. Similarly, the composition question for the approximate degree asks whether $\widetilde{\text{deg}}(f\circ g) = \tilde{\Theta}(\widetilde{\text{deg}}(f)\cdot\widetilde{\text{deg}}(g))$. These questions are two of the most important and well-studied problems in the field of analysis of Boolean functions, and yet we are far from answering them satisfactorily. 

It is known that the measures compose if one assumes various properties of the outer function $f$ (or inner function $g$). This paper extends the class of outer functions for which $\text{R}$ and $\widetilde{\text{deg}}$ compose. 

A recent landmark result (Ben-David and Blais, 2020) showed that $\text{R}(f \circ g) = \Omega(\text{noisyR}(f)\cdot \text{R}(g))$. This implies that composition holds whenever $\text{noisyR}(f) = \Tilde{\Theta}(\text{R}(f))$. We show two results:
\begin{itemize}
    \item[1.] When $\text{R}(f) = \Theta(n)$, then $\text{noisyR}(f) = \Theta(\text{R}(f))$. In other words, composition holds whenever the randomized query complexity of the outer function is full.
    
    \item[2.] If $\text{R}$ composes with respect to an outer function, then $\text{noisyR}$ also composes with respect to the same outer function.
\end{itemize}
On the other hand, no result of the type $\widetilde{\text{deg}}(f \circ g) = \Omega(M(f) \cdot \widetilde{\text{deg}}(g))$ (for some non-trivial complexity measure $M(\cdot)$) was known to the best of our knowledge.
We prove that 
\[
\widetilde{\text{deg}}(f\circ g) = \widetilde{\Omega}(\sqrt{\text{bs}(f)} \cdot \widetilde{\text{deg}}(g)),
\]
where $\text{bs}(f)$ is the block sensitivity of $f$.
This implies that $\widetilde{\text{deg}}$ composes when $\widetilde{\text{deg}}(f)$ is asymptotically equal to 
$\sqrt{\text{bs}(f)}$.

It is already known that both $\text{R}$ and $\widetilde{\text{deg}}$ compose when the outer function is symmetric. We also extend these results to weaker notions of symmetry with respect to the outer function.
\end{abstract}
\thispagestyle{empty}
\newpage
\tableofcontents
\newpage
\addtocounter{page}{-1}

\section{Introduction}
\label{sec: introduction}
For studying the complexity of Boolean functions, a number of simple complexity measures (like decision tree complexity, randomized query complexity, degree, certificate complexity and so on) have been studied over the years. (Refer to the survey~\cite{BW02} for an introduction to complexity measures of Boolean functions.) Understanding how these measures are related to each other \cite{ABK16, ABK+, ABBL+17, Huang}, and how they behave for various classes of Boolean functions has been a central area of research in complexity theory \cite{Paturi, Drucker11, Sun07}. 

A crucial step towards understanding a complexity measure is: how does the complexity measure behave when two Boolean functions are combined to obtain a new function (i.e., what is the relationship between the measure of the resulting function and the measures of the two individual functions)~\cite{BKT19, BGKW20, GSS16, Tal13}? One particularly natural combination of functions is called \emph{composition}, and it takes a central role in complexity theory.

For any two Boolean functions $f:\{0,1\}^n\to \{0,1\}$ and 
$g:\{0,1\}^m\to \{0,1\}$, the composed function $f \circ g:\zone^{nm} \to \zone$ is defined as the function  
\begin{align*}
    f \circ g (x_1, \dots ,x_n) = f(g(x_1), \dots ,g(x_n)),
\end{align*}
where $x_i \in \zone^m$ for $i \in [n]$.  For the function $f \circ g$, the function $f$ is called the outer function and $g$ is called the inner function. See Definition~\ref{defi: Generalized composition of functions} for a natural extension to partial functions. 

Let $M(\cdot)$ be a complexity measure of Boolean functions. A central question in complexity theory is the following.
\begin{question}[Composition question for $M$]
\label{question: composition question}
    Is the following true for all Boolean functions $f$ and $g$:
    \begin{align*}
        M(f \circ g) = \Tilde{\Theta}(M(f) \cdot M(g))?
    \end{align*}
\end{question}
The notation $\Tilde{\Theta}(\cdot)$ hides poly-logarithmic factors of the arity of the outer function $f$. 

Composition of Boolean functions with respect to different complexity measures is a very important and useful tool in areas like communication complexity, circuit complexity and many more. To take an example, a popular application of composition is to  create new functions demonstrating better separations (refer to~\cite{NS94, Tal13, Ambainis05, GSS16} for some other results of similar flavour). 

It is known that the answer to the composition question is in the affirmative for complexity measures like deterministic decision tree complexity~\cite{Savicky02,Tal13,Montanaro14}, degree~\cite{Tal13} and quantum query complexity~\cite{Reichardt11,LMR+11,Kim13}. 
While it is well understood how several complexity measures behave under composition, there are two important measures (even though well studied) for which this problems remains wide open: randomized query complexity (denoted by $\R$) and approximate degree (denoted by $\adeg$)~\cite{Sherstov12,NS94,Ambainis05,She13a,BT13,She13}. (See Definition~\ref{defi: randomized query complexity} and Definition~\ref{defi: approximate degree} for their respective formal definitions.)


For both $\R$ and $\adeg$ the upper bound inequality is known, i.e.,  $\R(f\circ g) = \Tilde{O}(\R(f)\cdot\R(g))$ (folklore) and $\adeg(f\circ g) = O(\adeg(f)\cdot\adeg(g))$ \cite{She13b}. Thus it is enough to prove the lower bound on the complexity of composed function in terms of the individual functions. Most of the attempts to prove this direction of the question have focused on special cases when the outer function has certain special properties\footnote{We note that some works have also studied the composition of $\R$ and $\adeg$ when the inner functions have special properties, for example, \cite{ABK16,BK18,AGJKLMSS17,GLSS19,Li21,BBGM22}.}. 

The initial steps taken towards answering the composition question for $\R$ were to show a lower bound by using a weaker measure of outer function than the randomized query complexity. In particular, it was shown that $\R(f \circ g) = \Omega(\s(f) \cdot \R(g))$~\cite{GJPW18, AKK16}, where $\s(f)$ is the sensitivity of $f$ (Definition~\ref{defi: block sensitivity}).
Since $\s(f) = \Theta(\R(f))$ for any symmetric function\footnote{Functions that depend only on the Hamming weight of their input.} $f$, these results show that $\R$ composes when the outer function is a symmetric function (like $\OR$, $\AND$, Majority, Parity, etc.). The lower bound was later improved to obtain  $\R(f \circ g) = \Omega(\fbs(f) \cdot \R(g))$ \cite{BDGHMT20, BB20}, where $\fbs(f)$ is the fractional block sensitivity of $f$ (Definition~\ref{defi: fbs}). 

Unfortunately, at this stage, there could even be a cubic gap between $\R$ and $\fbs$; the best known bound is $\R(f) = O(\fbs(f)^3)$ \cite{ABK+}. Given that there are already known Boolean functions with quadratic gap between $\fbs(f)$ and $\R(f)$ (e.g., $\mathsf{BKK}$ function \cite{ABK16}), it is natural to consider composition question for randomized query complexity when $\R$ is \emph{big} but $\fbs$ is \emph{small}. We take a step towards this problem by showing that composition for $\R$ holds when the outer function has full randomized query complexity, i.e., $\R(f)=\Theta(n)$, where $n$ is the arity of the outer function $f$. 



For composition of $\adeg$, Sherstov~\cite{Sherstov12} already showed that $\adeg (f\circ g)$ composes when the approximate degree of the outer function $f$ is $\Theta(n)$, where $n$ is the arity of the outer function. Thus approximate degree composes for several symmetric functions (having approximate degree $\Theta(n)$, like Majority and Parity). Though, until recently it was not even clear if $\adeg(\OR\circ\AND)=\Omega(\adeg(\OR)\,\adeg(\AND))$ 
(arguably the simplest of composed functions). $\OR$ has approximate degree $O(\sqrt{n})$, and thus the result of~\cite{Sherstov12} does not prove $\adeg$ composition when the outer function is $\OR$ (similarly for $\AND$). In a long series of work by \cite{NS94,Ambainis05,She13a,BT13,She13}, the question 
was finally resolved; it was later generalized to the case  when the outer function is symmetric~\cite{BBG+18}. 

In contrast to $\R$ composition, no lower bound on the approximate degree of composed function is known with a weaker measure on the outer function. It is well known that for any Boolean function $f$, $\sqrt{\s(f)} \leq \sqrt{\bs(f)} = O(\adeg(f))$ \cite{NS94}. So a natural step towards proving $\adeg$ composition is: prove a lower bound on $\adeg(f\circ g)$ by $\sqrt{\bs(f)}\cdot \adeg(g)$. We show this result by generalizing the approach of \cite{BBG+18}.

While the techniques used for the composition of $\R$ and $\adeg$ are quite different, one can still observe similarities between the classes of outer functions for which the composition of $\R$ and $\adeg$ is known to hold respectively. We further expand these similarities, by extending the classes of outer functions for which the composition theorem hold. 

\section{Our Results and Techniques}

It is well-known, by amplification, that $\R(f\circ g) = O(\R(f)\cdot \R(g) \cdot \log \R(f))$. In the case of approximate degree, Shrestov~\cite{She13b} showed that $\adeg(f\circ g) = O(\adeg(f)\cdot \adeg(g))$. So, to answer the composition question for $\R$ (or $\adeg$), we are only concerned about proving a lower bound on $\R(f\circ g)$ (or $\adeg(f\circ g)$) in terms of $\R(f)$ and $\R(g)$ (or $\adeg(f)$ and $\adeg(g)$) respectively. 

We split our results into three parts. In the first part we prove the tight lower bound on $\R(f\circ g)$ when the outer function has full randomized complexity. In the second part we provide a tight lower bound on $\adeg(f\circ g)$ in terms of $\bs(f)$ and $\adeg(g)$. Our results on the lower bound of $\R(f\circ g)$ and $\adeg(f\circ g)$ are summarized in Table~\ref{tab:complexity}.
Finally, we also prove composition theorems for $\R$ and $\adeg$ when the outer functions have a slightly relaxed notion of symmetry. 

\begin{table}
    \centering
    \scalebox{1.0}{
    \begin{tabular}{|c||c|c|}
        \hline
     &  {In terms of $\bs(f)$} &  {In terms of arity of $f$}  \\
\hline
\vspace{-10pt}
&&\\
 $\R$      & $\R(f \circ g) =\Tilde{\Omega}( \bs(f)\cdot \R(g))$ & $\R(f \circ g) = \Tilde{\Omega}(\R(f)\cdot \R(g))$ when $\R(f) = \Theta(n)$  \\
 & \cite{GJPW18} & Theorem~\ref{theo: composition with full RIntro}\\
\hline
\vspace{-10pt}
&&\\
 $\adeg$               & $\adeg(f \circ g) = \Tilde{\Omega}( \sqrt{\bs(f)}\cdot \adeg(g))$ &  $\adeg(f \circ g) = \Tilde{\Omega}( \adeg(f)\cdot \adeg(g))$ when $\adeg(f) = \Theta(n)$ \\
 & Theorem~\ref{thm: blocksensitivity composition} & \cite{Sherstov12}\\

\hline
    \end{tabular}}
    \caption{Composition of $\R$ and $\adeg$ depending on the complexity of the outer function in terms of block-sensitivity and arity.}
    \label{tab:complexity}
    \end{table}

\subsection{Lower bounds on $\R(f\circ g)$ when the outer function has full randomized query complexity}

Sherstov~\cite{Sherstov12} proved that $\adeg(f\circ g) = \Omega(\adeg(f)\cdot \adeg(g))$ when the approximate degree of the outer function $f$ is $\Theta(n)$, where $n$ is the arity of $f$. Though, a corresponding result for the case of randomized query complexity was not known. Our main result is to prove the corresponding theorem for randomized query complexity.
\begin{restatable}{theorem}{RandomizedMainThm}
\label{theo: composition with full RIntro}
Let $f$ be a partial Boolean function on $n$-bits such that $\R(f) = \Theta(n)$. Then for all partial functions $g$, we have
\[
\R(f \circ g) = {\Omega}(\R(f)\cdot \R(g)).
\]
\end{restatable}
The proof of this theorem is given in Section~\ref{section: Composition for R}.
Notice, since $\R(f \circ g) = O(\R(f)\cdot \R(g) \log \R(f))$ (by error reduction), Theorem~\ref{theo: composition with full RIntro} says that composition of $\R$ holds when the randomized query complexity of the outer function, $f$, is $\Theta(n)$. Next, we give main ideas behind the proof of the above theorem.

\paragraph*{Ideas behind proof of Theorem~\ref{theo: composition with full RIntro}} 
A crucial complexity measure that we use for the proof of Theorem~\ref{theo: composition with full RIntro} is called the \textit{noisy randomized query complexity}, first introduced by Ben-David and Blais~\cite{bb20focs} in order to study randomized query complexity.
Noisy randomized query complexity can be seen as a generalization of randomized query complexity where the algorithm can query a bit with any bias and only pays proportionally to the square of the bias in terms of cost
(see Definition~\ref{defi: Noisy Oracle Model of Computation}).
They give the following characterization of $\noisyR(f)$ (the noisy randomized query complexity of $f$).



\begin{restatable}[Ben-David and Blais~\cite{bb20focs}]{theorem}{bb20thm2}
\label{thm: bdb20 characterization of noisyR}
For all partial functions $f$ on $n$-bits, we have
\begin{equation}\label{Eq:Noisy}
    \noisyR(f) = \Theta\left(\frac{\R(f \circ \GapMaj_n)}{n} \right),
\end{equation}
where $\GapMaj_n$ is the Gap-Majority function on $n$ bits whose input is promised to have Hamming weight either $(n/2 + 2\sqrt{n})$ (in which case the output is $-1$) or  $(n/2 - 2\sqrt{n})$ (in which case the output is $1$).
\end{restatable}

We want to point out that the arity of $f$ and Gap-Majority is the \textit{same} in Theorem~\ref{thm: bdb20 characterization of noisyR}. Towards a proof of Theorem~\ref{theo: composition with full RIntro}, we first make the following crucial observation.

\begin{restatable}[]{observation}{GeneralizationNoisyR}
\label{our approach restated introduction}
Let $f$ be a partial Boolean function on $n$ bits. If $t(n) \geq 1$ is a non-decreasing function of $n$ and 
\begin{align*}
   \noisyR(f) =  {\Omega}\left(\frac{\R(f \circ \GapMaj_{t(n)})}{t(n)} \right),
\end{align*}
then $\R(f \circ g) = {\Omega}((\R(f) \cdot \R(g))/t(n))$ for all partial functions $g$. 
\end{restatable}


In particular choosing $t(n)$ to be $(\log n)$, if the outer function $f$ satisfies
\begin{align}
    \noisyR(f) = {\Omega}\left(\frac{\R(f \circ \GapMaj_{\log n})}{\log n} \right). \label{eqn: condition for composition}
\end{align}
 then the above observation gives $\R(f \circ g) = \Omega((\R(f) \cdot \R(g))/(\log n))$ for all partial functions $g$.

The Observation~\ref{our approach restated introduction} allows us to approach the composition question for randomized query complexity in a conceptually fresh manner. The goal of proving that randomized query complexity composes for a function or a class of functions, say upto $(\log n)$-factor, reduces to showing that Equation~\ref{eqn: condition for composition} holds for that function or class of functions for $t(n) = \log n$.

We are able to show that Equation~\ref{eqn: condition for composition} holds for all non-decreasing functions $t(n)$ with a slight overhead.
\begin{restatable}[]{theorem}{GeneralizedCharNoisyR}
\label{thm: bdb20 main theorem Intro}
Let $f$ be a partial function on $n$ bits and let $t \geq 1$, then
$\R(f \circ \GapMaj_t) = O\left(t \cdot \noisyR(f) + n \right)$.
\end{restatable}

Notice that this is a generalization of Ben-David and Blais' characterization of $\noisyR$ given by Theorem~\ref{thm: bdb20 characterization of noisyR} 
in one direction. 
To give an idea of the proof, their characterization (Theorem~\ref{thm: bdb20 characterization of noisyR}) shows that any noisy oracle algorithm for $f$ can be simulated using only two biases, $1$ and $1/\sqrt{n}$ (where $n$ is the arity of $f$), with only constant overhead. We generalize this by showing that the same simulation works with a slight overhead even when the bias $1/\sqrt{n}$ is replaced by a bias $1/\sqrt{t}$, for some $t \geq 1$. A proof the above theorem is provided in Appendix~\ref{section: proof of BB20 main thm}.

This seemingly inconsequential generalization allows us to complete the proof of Theorem~\ref{theo: composition with full RIntro}, i.e. if for an $n$-bit partial function $f$, $\R(f) = \Theta(n)$, then $\R(f \circ g) = \Tilde{\Theta}(\R(f) \cdot \R(g))$ for all partial functions $g$ (see Section~\ref{section: Composition for R} for details).

Furthermore, Theorem~\ref{thm: bdb20 main theorem Intro} even sheds light on \textit{the composition question for $\noisyR$.}
A corollary of this theorem is that if $\R$ composes with respect to an outer function, then $\noisyR$ also composes with respect to the same outer function (see Section~\ref{section: Composition for R} for a proof).

\begin{restatable}{corollary}{CompositionRImpliesNoisyR}
\label{theo: composition of R implies that of noisyR Intro}
Let 
$f$ be a partial Boolean function. If $\R(f \circ g) = \Tilde{\Theta}(\R(f)\cdot \R(g))$ for all partial functions $g$
then $\noisyR(f \circ g) = \Tilde{\Theta}(\noisyR(f) \cdot \noisyR(g))$.
\end{restatable}

\subsection{Lower bound on $\adeg(f\circ g)$ in terms of block sensitivity of $f$ and $\adeg(g)$}


As discussed in the introduction, the composition question for $\adeg$ is only known to hold when the outer function $f$ is symmetric \cite{BBG+18} or has high approximate degree \cite{Sherstov12}. 
There are also no known lower bounds on $\adeg(f\circ g)$ in terms of weaker measures of $f$ and $\adeg(g)$. Compare this with the situation with respect to composition of $\R$.  It was shown in~\cite{GJPW18} that $\R(f \circ g) = {\Omega}(\s(f)\,\R(g))$, where $\s(f)$ denotes the sensitivity of $f$. This was later strengthened to $\Omega(\fbs(f)\,\R(g))$ \cite{BDGHMT20, BB20}, where $\fbs(f)$ is the fractional block sensitivity of $f$. 

In this second part we show analogous lower bounds on approximate degree of composed function $f\circ g$. Our main result here is the following. 

\begin{restatable}[]{theorem}{ApproxDegBSMain}
\label{thm: blocksensitivity composition}
For all non-constant (possibly partial)\footnote{For definitions of block sensitivity and approximate degree in the context of partial functions, please see Definitions \ref{defi: block sensitivity} and \ref{defi: bounded degree}.} Boolean functions $f: \zone^n \to \zone$ and $g:\zone^m\to\zone$, we have 
\[\adeg(f \circ g) = \widetilde \Omega(\sqrt{\bs(f)}\cdot \adeg(g)).\]
\end{restatable}
We first note that the above theorem is tight in terms of block sensitivity, i.e., we cannot have $\adeg(f \circ g) = \widetilde \Omega(\bs(f)^c\cdot \adeg(g))$ for any $c > 1/2$.   This is because the $\OR$ function over $n$ bits witnesses the tight quadratic separation between $\adeg$ and $\bs$, i.e., $\adeg(\OR_n) = \Theta(\sqrt{n})= \Theta(\sqrt{\bs(\OR_n)})$ \cite{NS94}. 

We also get the following composition theorem as a corollary. It says that the composition for $\adeg$ holds 
when the outer function has minimal approximate degree with respect to its block sensitivity. Recall, $\adeg(f) = \Omega(\sqrt{\bs(f)})$ \cite{NS94}.  
\begin{restatable}[]{corollary}{ApproxDegBSCorro}
\label{cor:adeg-minimal-bs}
For all Boolean function $f: \zone^n \to \zone$ with $\adeg(f) = \Theta(\sqrt{\bs(f)})$ and for all $g:\zone^m\to\zone$, we have 
$\adeg(f \circ g) = \widetilde{\Theta}(\adeg(f)\cdot \adeg(g))$.
\end{restatable}
This complements a result of Sherstov \cite[Theorem 6.6]{Sherstov12},   which shows that composition of $\adeg$ holds when the outer function has maximal $\adeg$ 
with respect to its arity.

We further note that Corollary~\ref{cor:adeg-minimal-bs} covers new set of composed functions $f\circ g$ for which the composition theorem for $\adeg$ doesn't follow from the known results \cite{BBG+18,Sherstov12}. 
For example, consider the Rubinstein function $\rub$ with arity $n$ (Definition~\ref{defi: Rubinstein function}) as the outer function $f$. 
It is clearly not a symmetric function. 
It also doesn't have high approximate degree, i.e., $\adeg(\rub) = O(\sqrt{n}\log n)$ (Lemma~\ref{claim: ub adeg sink and rub}). Therefore, the composition of $\adeg(\rub\circ g)$ doesn't follow from the existing results. However, it follows from Corollary~\ref{cor:adeg-minimal-bs}, since $\bs(\rub)=\Omega(n)$ and so $\adeg(\rub)=\Tilde\Theta(\sqrt{\bs(\rub)})$. 

Another example is the sink function $\sink$ over $\binom{n}{2}$ variables (Definition~\ref{defi: sink}), which is also not a symmetric function. Furthermore, its approximate degree is $O(\sqrt{n}\log n)$ (Lemma~\ref{claim: ub adeg sink and rub}).  Therefore, the composition of $\adeg(\sink\circ g)$ also doesn't follow from the existing results. 
Again, it follows from Corollary~\ref{cor:adeg-minimal-bs}, since 
$\bs(\sink) = \Theta(n)$ (Observation~\ref{obs: bs of sink}) and $\adeg(\sink) = \Tilde\Theta(\sqrt{n})$.  

\paragraph*{Ideas behind proof of Theorem~\ref{thm: blocksensitivity composition}}
We will first sketch the proof ideas in the case when $f$ and $g$ are total Boolean functions, and then explain how to extend it to partial functions too. 

Our starting point is the well known Nisan-Szegedy's embedding of $\PrOR$ over $\bs(f)$ many bits in a Boolean function $f$ \cite{NS94}. Carrying out this transformation in $f\circ g$ embeds $\PrOR_{\bs(f)} \circ (g_1,\ldots ,g_{\bs(f)})$ into $f\circ g$, where $g_1,\ldots ,g_{\bs(f)}$ are different partial functions such that $\bdeg(g_i)\geq \adeg(g)$ for all $i \in [\bs(f)]$\footnote{$\bdeg$ is the notion of approximate degree in the context of partial functions. For a formal definition, see Definition~\ref{defi: bounded degree}.}. Since the transformation is just substitutions of variables by constants, we further have 
\begin{align}
\label{eq:adeg-intro}
\adeg(f\circ g) \geq \bdeg(\PrOR_{\bs(f)} \circ (g_1,\ldots ,g_{\bs(f)})).
\end{align}
It now looks like that we can appeal to the composition theorem for $\PrOR$ (Theorem~\ref{thm:promise-or-composition}) \cite{BBG+18} to obtain our theorem. However, there is a technical difficulty -- Theorem~\ref{thm:promise-or-composition} doesn't hold for \emph{different} inner \emph{partial} functions. 
It only deals with a single total inner function. We therefore generalize the proof of Theorem~\ref{thm:promise-or-composition} to obtain the following general version of the composition theorem for $\PrOR$.
\begin{restatable}[]{theorem}{ApproxDegPrORThm}
\label{thm: for different function in PrOR}
For any partial Boolean functions $g_1, g_2, \dots, g_n$, we have
\begin{align*}
\bdeg\left(\PrOR_n\circ (g_1,g_2,\ldots ,g_n)\right) = \Omega\left(\frac{\sqrt{n}\cdot \min_{i=1}^n\bdeg(g_i)}{\log n}\right).
\end{align*}
\end{restatable}
We can now obtain our main theorem from Eq.~\eqref{eq:adeg-intro} and Theorem~\ref{thm: for different function in PrOR}. The proof of Theorem~\ref{thm: for different function in PrOR} is a generalization of Theorem~\ref{thm:promise-or-composition}. For lack of space, we present it in Appendix~\ref{appendix: shifted OR}. 

We end this part with a comment on how to extend our main theorem to partial functions. 
Note that with the appropriate definition of block sensitivity (Definition~\ref{defi: block sensitivity}), the embedding of Nisan-Szegedy's carries through and then the rest of the proof is identical.

\subsection{Composition results when the outer functions has some symmetry}

The class of symmetric functions capture many important function like $\OR$, $\AND$, Parity and Majority. 
Recall that a function is symmetric when the function value only depends on the Hamming weight of the input; in other words, a function is symmetric iff its value on an input remains unchanged even after permuting the bits of the input. As noted earlier, both for $\R$ and $\adeg$, composition was known to hold when the outer function was symmetric.


A natural question is, whether one can prove composition theorems when the outer function is \emph{weakly} symmetric (it is symmetric with respect to a weaker notion of symmetry). In this paper we consider one such notion of symmetry -- junta-symmetric functions.

\begin{defi}[$k$-junta symmetric function]
\label{junta symmetric function}
A function $f: \zone^n \to \zone$ is called a $k$-junta symmetric function if there exists a set $\mathcal{J}$ of size $k$ of variables such that the function value depends on assignments to the variables in $\mathcal{J}$ as well as on the Hamming weight of the whole input.
\end{defi}


$k$-junta symmetric functions can be seen as a mixture of symmetric functions and $k$-juntas. This class of functions has been considered previously in literature, particularly in \cite{CFGM12, BWY15} where these functions plays a crucial role. \cite{CFGM12} even presents multiple characterisations of $k$-junta symmetric functions for constant $k$. 
Note that by definition an arbitrary $k$-junta (i.e., a function that depend on $k$ variables) is also a $k$-junta symmetric function, since we can consider the dependence on Hamming weight to be trivial. Thus, this notion loses out on the symmetry of the function considered. We, therefore, consider the class of \emph{strongly} $k$-junta symmetric functions. 
\begin{defi}[Strongly $k$-junta symmetric function]
\label{strongly junta symmetric function}
A $k$-junta symmetric function is called strongly $k$-junta symmetric if every variable is influential. In other words, 
there exists a setting to the junta variables such that 
the function value depends on the Hamming weight of the whole input in a non-trivial way. 
\end{defi}

We prove that if the outer function is strongly $\sqrt{n}$-junta symmetric (``strongly'' indicating that the dependence on the Hamming weight is non-trivial) then $\adeg$ composes. On the other hand, Theorem~\ref{theo: composition with full RIntro} implies that $\R$ composes for any strongly $k$-junta symmetric functions (as long as $n-k = \Theta(n)$).
\begin{theorem}
\label{thm: junta sym composition Intro}
For any strongly $k$-junta symmetric function $f: \zone^n \to \zone$ and any Boolean function $g: \zone^m \to \zone$, we have 
\begin{itemize}
    \item $\adeg(f\circ g) = \widetilde{\Theta}({\adeg(f)\cdot \adeg(g)})$ where $ k = O(\sqrt{n})$.
    \item $\R(f\circ g) = \widetilde{\Theta}({\R(f)\cdot \R(g)})$ where $n-k = \Theta(n)$.
\end{itemize}
\end{theorem}


For the lack of space, the proof of the above theorem is given in Appendix~\ref{appendix: junta symmetric function}. Note that if one is able to prove the above theorem for $k$-junta-symmetric functions (without the requirement of ``strongly'') for any non-constant $k$ then we would have the full composition theorem. 




\section*{Organization of the paper}

We have formally defined complexity measures and Boolean functions needed for our results in Section~\ref{section: prelims}. Section~\ref{section: Composition for R} contains proofs of our results related to the composition of randomized query complexity (Theorem~\ref{theo: composition with full RIntro}). In Section~\ref{section: block sensitivity proof} we give the proof of our result for the composition of approximate degree (Theorem~\ref{thm: blocksensitivity composition}). Finally, the results about composition of functions with weak notion of symmetry are in  Appendix~\ref{appendix: junta symmetric function}.



\section{Preliminaries}
\label{section: prelims}

\textbf{Notations:}
We will use $[n]$ to represent the set $\{1, \dots, n\}$.  For any (possibly partial) Boolean function $f:\zone^n \to \{0,1,*\}$ we will denote by $\Dom(f)$ the set $f^{-1}(\{0,1\})$. The arity of $f$ is the number of variables - in this case $n$. A Boolean function $f:\zone^n \to \{0,1,*\}$ is said to be total if $\Dom(f) = \{0,1\}^n$. Any function (not otherwise stated) will be a total Boolean function.

For any $x\in \{0,1\}^n$, we will use $|x|$ to denote the number of $1$s in $x$, that is, the Hamming weight of the string $x$. The string $x^i$ denotes the modified string $x$ with the $i$-th bit flipped. Similarly, $x^{B}$ is defined to be the string such that all the bits whose index is contained in the set $B \subseteq [n]$ are flipped in $x$. 

Following is a formal definition of (partial) function composition. 
\begin{defi}[Generalized composition of functions]
\label{defi: Generalized composition of functions}
For any (possibly partial) Boolean function $f:\zone^n \to \{0,1,*\}$ and $n$ (possibly partial) Boolean functions $g_1, g_2, \dots, g_n$, define the (possibly partial) composed function 
\[f \circ (g_1, g_2, \dots, g_n)(x_1,x_2, \dots, x_n)= f(g_1(x_1),g_2(x_2), \dots, g_n(x_n)),\]
where $g_i$'s can have different arities and, moreover,  if $x_i \notin \Dom(g_i)$ for any $i\in [n]$ or the string $(g_1(x_1),g_2(x_2), \dots, g_n(x_n)) \notin \Dom(f)$, then the function $f \circ g$ outputs $*$.
\end{defi}





In this paper we use the standard definitions of various complexity measures like randomized query complexity,  sensitivity, block-sensitivity, fractional block sensitivity, noisy randomized query complexity, approximate degree and bounded approximate degree and the partial function Promise-$\OR$. We present the formal definitions in the following subsections. 

\ 
\subsection{Standard definitions and functions}
\label{appendix: pewlim}
\subsubsection{Standard definition of complexity measures}
We look at many different complexity measures in the paper, let us start with the formal definition of $\R$ and $\adeg$. 

\begin{defi}[Randomized query complexity $(\R)$]
\label{defi: randomized query complexity}
 Let $f : \zone^n \to \{0,1,*\}$ be a (possibly partial) Boolean function. A randomized query algorithm $A$ computes $f$ if 
 $\forall x \in \Dom(f), \Pr[A(x) \neq f(x)] \leq 1/3$, 
 where the probability is over the internal randomness of the algorithm. The cost of the algorithm $A$, \text{cost}(A), is the number of queries made in the worst case over any input as well as internal randomness. The randomized query complexity of $f$, denoted by $\R(f)$, is defined as
 
 \begin{align*}
     \R(f) = \min_{A \text{ computes } f} \text{cost}(A).
 \end{align*}
\end{defi}


\begin{defi}[Approximate degree $(\adeg)$]
\label{defi: approximate degree}
A polynomial $p: \mathbb{R}^n \to \mathbb{R}$ is said to approximate a Boolean function $f: \zone^n \to \{0,1\}$ if 
 \(   | p(x) -f(x) | \leq 1/3, \quad \forall x \in \zone^n\).
The approximate degree of $f$, $\adeg(f)$, is the minimum possible degree of a polynomial which approximates $f$.
\end{defi}
Note that the constant $1/3$ in the above definitions can be replaced by any constant strictly smaller than $1/2$ which changes $\adeg(f)$ by only a constant factor. 

Other than $\R$ and $\adeg$, two important related measures are sensitivity ($\s(f)$) and block sensitivity ($\bs(f)$). While the sensitivity and block sensitivity of a total function is well defined, we note that for the case of partial functions there are at least two valid ways of extending the definition from total functions to partial functions. All our results in this paper will hold for partial functions with the following definitions of sensitivity and block sensitivity. 

\begin{defi}
\label{defi: block sensitivity}
The sensitivity $\s(f,x)$ of a function $f: \zone \to \zonep$ on $x$ is the maximum number $s$ such that there are indices $i_1, i_2, \dots, i_s \in [n]$ with
$f(x^{i_j}) = 1-f(x)$, for all $1\leq j\leq s$. 
Here $x^{i}$ is obtained from $x$ by flipping the $i^{th}$ bit. 
The sensitivity of $f$ is defined to be $\s(f) = \max_{x \in \Dom(f)} \s(f,x)$.

Similarly, the block sensitivity $\bs(f,x)$ of a function $f: \zone \to \zonep$ on $x$ is the maximum number $b$ such that there are disjoint sets $B_1, B_2, \dots, B_b \subseteq [n]$ with $f(x^{B_j}) = 1 - f(x)$ for all $1\leq j\leq b$.
Recall $x^{B_j}$ is obtained from $x$ by flipping all bits inside $B_j$.
The block sensitivity of $f$ is defined to be $\bs(f) = \max_{x \in \Dom(f)} \bs(f,x)$.
\end{defi}
In the definition of block sensitivity, the constraint that the blocks has to be disjoint can be relaxed by extending the definition to ``fractional blocks''. This gives the measure of fractional block sensitivity. 
\begin{defi}
\label{defi: fbs}
The fractional block sensitivity $\fbs(f,x)$ of a function $f: \zone \to \zonep$ on $x$ is the maximum value of $\sum_{j=i}^{b}p_j$ such that there are sets $B_1, B_2, \dots, B_b \subseteq [n]$ and $p_1, \dots, p_b\in (0,1]$ satisfying the following two conditions.
\begin{itemize}
    \item For each $1\leq j \leq b$, $f(x^{B_j}) = 1 - f(x)$, and
    \item For each $1\leq i \leq n$,  $\sum_{j\colon i \in B_j} p_j \leq 1$.
\end{itemize}
The fractional block sensitivity of $f$ is defined to be $\fbs(f) = \max_{x \in \Dom(f)} \fbs(f,x)$.
\end{defi}

\subsubsection{For the composition of $\R$}
\label{sec:prelimR}

The function Gap-Majority has played an important role in the study of composition of $\R$.

\begin{defi}[Gap-Majority]
\label{defi: gapmaj}
The function $\GapMaj_t : \zone^t \to \{0,1,*\}$ is a partial function with arity $t$ such that 
\begin{align*}
    \GapMaj_t(x) = \begin{cases}
    1 & \text{if } |x| = t/2 + 2\sqrt{t},\\
    0 & \text{if } |x| = t/2 - 2\sqrt{t}, \\
    * & \text{otherwise.}
    \end{cases}
\end{align*}
\end{defi}

It can be shown that $\R(\GapMaj_t) = \Theta(t)$~\cite{bb20focs}.

In regards to the composition question of $\R$, one of the most significant complexity measures (defined by Ben-David and Blais~\cite{bb20focs}) is that of $\noisyR$. We first define the noisy oracle model.

\begin{defi}[Noisy Oracle Model and Noisy Oracle Access to a String (\cite{bb20focs})]
For $b \in \zone$, a noisy oracle to $b$ takes a parameter $-1 \leq \gamma \leq 1$ as input and returns a bit $b'$ such that $\Pr[b' = b] = (1+\gamma)/2$. 
The cost of one such query is $\gamma^2$. Each query to noisy oracle returns independent bits.

For $x = (x_1, \dots, x_n) \in \zone^n$, noisy oracle access to $x$ is access to $n$ independent noisy oracles, one for each bit $x_i$, $i \in [n]$.
\end{defi}

Next, we define the noisy oracle model of computation.

\begin{defi}[Noisy Oracle Model of Computation (\cite{bb20focs})]
 \label{defi: Noisy Oracle Model of Computation}
 Let $f : \zone^n \to \{0,1,*\}$ be a partial Boolean function. A $\noisyR$ query algorithm $A$ computes $f$ if for all $x \in \Dom(f)$, $\Pr[A(x) \neq f(x)] \leq 1/3$, where $A$ is a randomized algorithm given noisy oracle access to $x$, and the probability is over both noisy oracle calls and the internal randomness of the algorithm $A$. The cost of the algorithm $A$ for an input $x$ is the sum of the cost of all noisy oracle calls made by $A$ on $x$, and the cost of $A$, \text{cost}(A), 
 is the maximum cost over all $x \in \Dom(f)$. The $\noisyR$ randomized query complexity of $f$, denoted by $\noisyR(f)$, is defined as
 
 \begin{align*}
     \noisyR(f) = \min_{A \text{ computes } f} \text{cost}(A).
 \end{align*}
\end{defi}

Again, $1/3$ in the above definition can be replaced by any constant $<1/2$. 
If only queries with $\gamma = 1$ are allowed in the noisy query model, then we obtain the usual randomized algorithm for $f$, thus $\noisyR(f) = O(\R(f))$.   


\textbf{For the composition of $\adeg$}
\label{sec:prelimadeg}
The definition of $\adeg$ can naturally be extended to partial functions $f$ by restricting the definition to hold only for inputs in $\Dom(f)$, i.e., it doesn't specify the value of the approximating polynomial on inputs \emph{not} in $\Dom(f)$. So the approximating polynomial can take arbitrarily large values on points outside the domain. However, for the purpose of understanding the composition of approximate degree of Boolean functions (or even total Boolean functions) one need a measure of approximate degree of partial Boolean functions which is bounded on all the points of the Boolean cube. 




\begin{defi}[Bounded approximate degree $(\bdeg)$]
\label{defi: bounded degree}
For a partial Boolean function $f:\{0,1\}^n \to \{0,1,*\}$, the bounded approximate degree of $f$ $(\bdeg(f))$ is the minimum possible degree of a polynomial $p$ such that 
\begin{itemize}
    \item $| p(x) -f(x) | \leq 1/3,  \quad \forall x \in \Dom(f)$, and
    \item $ 0 \leq p(x) \leq 1  \quad \forall x \in \zone^n$.
\end{itemize}
\end{defi}

In other words, we take the minimum possible degree of a polynomial which is bounded for all possible inputs ($p(x)\in [0,1]$ for all $x\in \zone^n$), and it approximates $f$ in the usual sense over $\Dom(f)$. 

Over the years people have tried to study the composition of $\adeg$ with different outer functions. 
In this context the following restriction of $\OR$ is an important partial function:

\begin{defi}[Promise-$\OR$]
\label{defi: promised OR}
Promise-$\OR$ (denoted by $\PrOR_n$) is the function $\PrOR_n : \zone^n \to \{0,1,*\}$ such that:
\begin{align*}
    \PrOR_n(x)= \begin{cases}
    0 & \textit{if } |x|=0,\\
    1 & \textit{if } |x|=1,\\
    * & \textit{otherwise.}
     \end{cases}
\end{align*}
\end{defi}

\noindent\textbf{Some useful previous results:} 
We will also be crucially using a few results from prior works in our proofs. The following are a couple of useful results on $\noisyR$.


\begin{lemma}[\cite{bb20focs}]
\label{lamma: noisyR omega 1}
Let $f$ be a non-constant partial Boolean function then $\noisyR(f) = \Omega(1)$.
\end{lemma}

\begin{theorem}[\cite{bb20focs}]
\label{thm: bdb20 composition at least noisyR times R}
For all partial functions $f$ and $g$, $\R(f \circ g) = \Omega(\noisyR(f)\cdot \R(g))$.
\end{theorem}

We will also be using the following theorem of \cite{BBG+18} regarding the composition question of $\bdeg$ when the outer function is $\PrOR_n$.  Informally, we will call it the Promise-$\OR$ composition theorem. 
\begin{theorem}[\cite{BBG+18}]
\label{thm:promise-or-composition}
For any Boolean function $g: \zone^m \to \{0,1\}$ we have,
\begin{align*}
    \bdeg\left(\PrOR_n \circ g\right) = \Omega\left(\sqrt{n}\cdot\adeg(g)/\log n\right).
\end{align*}
\end{theorem}
\section{Results about composition of $\R$}
\label{section: Composition for R}



This section is devoted to the results related to the composition of randomized query complexity. Our main result states that composition of $\R$ holds if the outer function has full randomized query complexity (Theorem~\ref{theo: composition with full RIntro}). As mentioned in the proof idea, the proof critically depends on the notion of noisy randomized query complexity and its properties (introduced by Ben-David and Blais~\cite{bb20focs}). 

Recall the definition of noisy randomized query complexity of a function $f$ from Definition~\ref{defi: Noisy Oracle Model of Computation}. As mentioned in the introduction (Theorem~\ref{thm: bdb20 characterization of noisyR}), Ben-David and Blais~\cite{bb20focs} proved that     
\begin{align}\label{Eq:Noisy1}
    \noisyR(f) = \Theta\left(\frac{\R(f \circ \GapMaj_n)}{n} \right),
\end{align}
where $\GapMaj_n$ is the Gap-Majority function on $n$ bits. Note that Ben-David and Blais proved Equation~\ref{Eq:Noisy1} when the arity of functions $f$ and Gap-Majority is the same.
We show that if Equation~\ref{Eq:Noisy1} can be generalized for Gap-Majority functions of arbitrary arity for some outer function $f$, then randomized query complexity composes for the function $f$. We restate the following observation from the introduction.


\GeneralizationNoisyR*


\begin{proof}
Suppose $\noisyR(f) = {\Omega}\left(\frac{\R(f \circ \GapMaj_{t(n)})}{t(n)} \right)$, since $\R(f \circ \GapMaj_t) \geq \R(f)$, we have $\noisyR(f) = {\Omega}(\R(f)/(t(n))$. Theorem~\ref{thm: bdb20 composition at least noisyR times R} implies that a lower bound on $\noisyR$ translates to a lower bound on $\R(f \circ g)$. We have,
\begin{align*}
\R(f \circ g)
&= \Omega(\noisyR(f) \cdot \R(g)) \tag{Theorem~\ref{thm: bdb20 composition at least noisyR times R}} \\
&= \Omega\left( \frac{\R(f) \cdot \R(g)}{t(n)} \right). \qedhere
\end{align*}
\end{proof}

Observation~\ref{our approach restated introduction} follows from the above observation by choosing $t(n)$ to be a small function of $n$.

We restate from Section~\ref{sec: introduction} our generalized characterization of $\noisyR$ (i.e., generalization of Equation~\ref{Eq:Noisy1}, see Appendix~\ref{section: proof of BB20 main thm} for its proof).

\GeneralizedCharNoisyR*

This allows us to show that if for an $n$-bit partial function $f$, $\R(f) = \Theta(n)$, then $\R(f \circ g) = \Tilde{\Theta}(\R(f) \cdot \R(g))$ for all partial functions $g$ (Theorem~\ref{theo: composition with full RIntro}).

The proof of Theorem~\ref{theo: composition with full RIntro} is discussed in Section~\ref{sec: proof of composition of R for full R}. A corollary of this theorem is that if $\R$ composes with respect to an outer function, then $\noisyR$ also composes with respect to the same outer function (Corollary~\ref{theo: composition of R implies that of noisyR Intro}).



We give proof of Theorem~\ref{theo: composition with full RIntro} in the next section and prove Corollary~\ref{theo: composition of R implies that of noisyR Intro} in Section~\ref{sec: Composition of R implies composition of noisyR}. We need the following theorem for these proofs, which lower bounds $R(f \circ g)$ in terms of $R(f)$ and $R(g)$.

\begin{theorem}[\cite{GLSS19}]
\label{thm: glss19 main thm}
Let $f$ and $g$ be partial functions then $\R(f \circ g) = \Omega(\R(f)\cdot \sqrt{\R(g)})$.
\end{theorem}





\subsection{Composition for functions with $\R(f) = \Theta(n)$}
\label{sec: proof of composition of R for full R}

We restate the theorem below.
\RandomizedMainThm*

\begin{proof}
    
From Theorem~\ref{thm: glss19 main thm} we have a lower bound on the randomized query complexity of $(f \circ \GapMaj_t)$:
\begin{align}
    \R(f \circ \GapMaj_t) = \Omega(\R(f) \cdot \sqrt{t}). 
\end{align}

On the other hand, Theorem~\ref{thm: bdb20 main theorem Intro} gives an upper bound of $ O\left(t \cdot \noisyR(f) + n \right)$ on $\R(f \circ \GapMaj_t)$.
Thus, choosing $t = \left(\frac{C \cdot n}{\noisyR(f)}\right)$ for a large enough constant $C$, we have
\begin{align*}
    \R(f) \cdot \sqrt{\frac{n}{\noisyR(f)}} = O\left(\frac{n}{\noisyR(f)} \cdot \noisyR(f) + n\right). 
\end{align*}

This implies that
\begin{align}
    \R(f) 
    &= O\left(\sqrt{n \cdot \noisyR(f)} \right). \label{eqn: observation R vs noisyR}
\end{align}
Thus, if $\R(f) = \Theta(n)$, then $\noisyR(f) = \Omega(\R(f))$, which implies composition from Theorem~\ref{thm: bdb20 composition at least noisyR times R}.
\end{proof}

Notice that Equation~\ref{eqn: observation R vs noisyR} is equivalent to the following observation.
\begin{observation}
\label{observation: R vs noisyR for general functions}
Let $f$ be a partial Boolean function on $n$-bits. Then,
$\noisyR(f) = \Omega\left( \frac{\R(f)^2}{n} \right)$.
\end{observation}
When $\R(f) = \Theta(n)$, we have already seen that Observation~\ref{observation: R vs noisyR for general functions} implies composition of randomized query complexity when the outer function is $f$.

Though, Observation~\ref{observation: R vs noisyR for general functions} implies a more general result. 
When $\R(f)$ is close to $n$ (arity of $f$), Observation~\ref{observation: R vs noisyR for general functions} places a limit on the 
gap between $\R(f)$ and $\noisyR(f)$ (consequently on the violation of composition with outer function being $f$). These implications are formally discussed in the next Section~\ref{sec: appendix implications of Proof of Observation}.


\subsection{Additional implications of Observation~\ref{observation: R vs noisyR for general functions}}
\label{sec: appendix implications of Proof of Observation}

Without loss of generality we can assume $\R(f \circ g) = \Omega(\R(g))$ (note that this is true when $f$ is non-constant).

Ben-David and Blais~\cite{bb20focs} gave a counterexample for composition, but the arity of the used function was very high compared to the randomized query complexity. They observed that the composition can still be true in the weaker sense:
\begin{align*}
    \R(f \circ g) = \Omega\left( \frac{\R(f) \cdot \R(g)}{\log n}\right).
\end{align*}

Observation~\ref{observation: R vs noisyR for general functions} shows that a much weaker composition result is true.

\begin{corollary}
\label{coro: weak composition}
Let $f$ and $g$ be partial functions on $n$ and $m$ bits respectively. If $\R(f \circ g) = \Omega(\R(g))$, then
\begin{align*}
    \R(f \circ g) = \Omega\left(\frac{\R(f)\cdot \R(g)}{\sqrt{n}}\right).
\end{align*}
\end{corollary}
\begin{proof}
\begin{align}
    \R(f \circ g) 
    &= \Omega(\noisyR(f) \cdot \R(g)) \tag{Theorem~\ref{thm: bdb20 composition at least noisyR times R}} \\
    &= \Omega\left(\frac{\R(f)^2 \cdot \R(g)}{n} \right). \label{eq: implication of observation123}
\end{align}
Where the last equality follows from Observation~\ref{observation: R vs noisyR for general functions}~\footnote{Sherstov~\cite{Sherstov12} proved that for Boolean functions $f$ and $g$, $\adeg(f \circ g) = \Omega((\adeg(f)^2 \adeg(g))/n)$. Thus in Equation~\ref{eq: implication of observation123} we prove the same result but in the randomized world.}
Now there are two cases:
\begin{itemize}
    \item \textbf{Case 1.} $\R(f) = O(\sqrt{n})$. In this case $\R(f)/\sqrt{n} = O(1)$ and since we assumed $\R(f \circ g) = \Omega(\R(g))$, the claim follows from Equation~\ref{eq: implication of observation123}.
    
    \item \textbf{Case 2.} $\R(f) = \Theta(n^{1/2} \cdot t(n))$ where $t(n)$ is a strictly increasing function of $n$. Thus,
    \begin{align*}
        \frac{\R(f)^2 \cdot \R(g)}{n} = \Omega\left(t(n)^2 \cdot \R(g)\right) = \Omega\left(\frac{\R(f) \cdot \R(g)}{\sqrt{n}} \right).
    \end{align*}
    Again, the claim follows from Equation~\ref{eq: implication of observation123}.
\end{itemize}
\end{proof}

The weaker composition, Corollary~\ref{coro: weak composition}, implies that if $\R(f)$ and $\R(g)$ are comparable to the arity of these functions, the randomized query complexity of $f\circ g$ is ``not far'' from the conjectured randomized query complexity $\R(f) \cdot \R(g)$. In other words, if there is a large polynomial separation between $\R(f \circ g)$ and $(\R(f) \cdot \R(g))$, then $\R(f)$ and $\R(g)$ can not be too large.


\begin{corollary}
\label{coro: sqrt n barrier}
Let $f$ and $g$ be partial functions such that $f$ is a function on $n$-bits and $g$ is a function on $t(n)$-bits where $t(n)$ is a strictly increasing function of $n$. If $\R(f) = \Theta(n^{\beta})$, $\R(g) = \Theta(n^{\gamma})$ and $\R(f \circ g) = O((\R(f) \cdot \R(g))^{\alpha})$, where $\alpha < 1$ is a constant, then $(1-\alpha)(\alpha+\beta) < 1/2$.
\end{corollary}
\begin{proof}
For some constants $A$ and $B$ we have
\begin{align*}
    A \cdot \frac{\R(f)\cdot \R(g)}{\sqrt{n}} \leq \R(f \circ g) \leq B \cdot (\R(f) \cdot  \R(g))^{\alpha},
\end{align*}
where the first inequality follows from Corollary~\ref{coro: weak composition} and second from assumption. Assigning the values of $\R(f)$ and $\R(g)$ in terms on $n$ we have,
\begin{align*}
    & A \cdot n^{\beta + \gamma - 1/2} \leq B \cdot n^{\alpha(\beta + \gamma)} \\
    & n^{(1-\alpha)(\beta + \gamma) - 1/2} \leq \frac{B}{A}.
\end{align*}
which implies, for large enough $n$, $(1-\alpha)(\beta + \gamma) \leq 1/2$.
\end{proof}

A special case of the above corollary is when arity and randomized query complexity of $g$ are superpolynomial in $n$. In this case a polynomial gap between $\R(f \circ g)$ and $(\R(f) \cdot \R(g)))$ is not possible.

Another implication of Theorem~\ref{theo: composition with full RIntro} is that composition of $\R$ for an outer function $f$ implies the composition of $\noisyR$ for outer function being $f$ (Corollary~\ref{theo: composition of R implies that of noisyR Intro}).

\subsection{Proof of Corollary~\ref{theo: composition of R implies that of noisyR Intro}}
\label{sec: Composition of R implies composition of noisyR}

First, we need the following lemma which follows from Theorem~\ref{thm: bdb20 main theorem Intro}, Theorem~\ref{thm: bdb20 composition at least noisyR times R} and Lemma~\ref{lamma: noisyR omega 1}.

\begin{lemma}
\label{lemma: implication 1 of R composition}
Let $f$ be a partial function on $n$ bits and let $t = \Omega(n)$. Then
\[
    \noisyR(f) = \Theta\left(\frac{\R(f \circ \GapMaj_t)}{t}\right).
\]
\end{lemma}
\begin{proof}
From Theorem~\ref{thm: bdb20 main theorem Intro} we have for all $t \geq 1$,
$\R(f \circ \GapMaj_t) = O( t\cdot\noisyR(f) + n)$.
Since we have assumed $t = \Omega(n)$ and $\noisyR(f) = \Omega(1)$ (Lemma~\ref{lamma: noisyR omega 1}), we get $\R(f \circ \GapMaj_t) = O(t \cdot \noisyR(f))$. Thus,
 $
 \noisyR(f) = \Omega\left(\frac{\R(f \circ \GapMaj_t)}{t}\right)$.

The upper bound $\noisyR(f) = O\left(\frac{\R(f \circ \GapMaj_t)}{t}\right)$ follows from Theorem~\ref{thm: bdb20 composition at least noisyR times R} and the fact that $\R(\GapMaj_t) = \Theta(t)$.
\end{proof}

Now we prove that if $\R$ composes for $f$ then $\noisyR$ composes for that $f$. For convenience, we recall the statement of the corollary from the introduction.

\CompositionRImpliesNoisyR*


\begin{proof}
From Theorem~\ref{thm: bdb20 characterization of noisyR}, we have
\begin{align*}
    \noisyR(f \circ g)
    &= \Theta\left(\frac{\R\left((f \circ g) \circ \GapMaj_{mn}\right)}{mn} \right).
\end{align*}
Since $(f \circ g) \circ h = f \circ (g \circ h)$, the right hand side of the above expression is equal to 
\begin{align*}
    \Theta\left(\frac{\R\left(f \circ (g \circ \GapMaj_{mn}\right))}{mn} \right).
\end{align*}

The proof follows from the assumption that $\R$ composes and Lemma~\ref{lemma: implication 1 of R composition}.
\begin{align*}
    \noisyR(f \circ g)
    &= \Theta\left(\frac{\R\left(f\right)\cdot \R\left(g \circ \GapMaj_{mn}\right)}{mn} \right) \tag{assuming $\R$ composes}\\
    &= \Theta\left(\R(f) \cdot \noisyR(g) \right) \tag{from Lemma~\ref{lemma: implication 1 of R composition}}\\
    &= \Theta\left(\noisyR(f)\cdot \noisyR(g) \right) \tag{assuming $\R$ composes}.
\end{align*}
\end{proof}

\section{Composition of approximate degree in terms of block sensitivity}
\label{section: block sensitivity proof}

In this section we study the composition question for approximate degree. Recall that the composition question asks: whether for all Boolean functions $f$ and $g$ 
\[\adeg(f\circ g) = \Tilde{\Omega}( \adeg(f)\,\adeg(g))?\]
Following our discussion from the introduction, we know that the above composition is known to hold for only two sub-classes of outer functions, namely symmetric functions \cite{BBG+18} and functions with high approximate degree \cite{Sherstov12}. It is thus natural to seek weaker lower bounds to make progress towards the composition question. One way to weaken the expression on the right-hand side would be to replace the measure $\adeg(f)$ by a weaker measure (like $\sqrt{\s(f)}$, $\sqrt{\bs(f)}$ or $\sqrt{\fbs(f)}$). Here we will establish one such lower bound of $\sqrt{\bs(f)}\,\adeg(g)$. 

We restate our theorem now. 
\ApproxDegBSMain*
We note that many analogous results are known in the setting of composition of $\R$; see, for example, \cite{GJPW18,BDGHMT20,BB20,BK18,AGJKLMSS17,GLSS19,BBGM22}. 
To the best of our knowledge, this is the first such result in the setting of $\adeg$. 
We present only a proof sketch here; most of the technical parts of the proof appear in Appendix~\ref{appendix: shifted OR}. 

Further, we present the sketch of the proof in two parts. For simplicity, in the first part we sketch a proof of the lower bound $\sqrt{\s(f)}\,\adeg(g)$ for total function $f$, and then in the second part we modify the arguments to obtain Theorem~\ref{thm: blocksensitivity composition}. 

We begin with a proof sketch for a lower bound of 
$\sqrt{\s(f)}\,\adeg(g)$. Let $x\in \{0,1\}^n$ be an input having the maximum sensitivity with respect to  $f$, and $S\subseteq [n]$ be the set of sensitive bits at $x$ ($|S| = \s(f)$). Consider the subfunction $f'$ obtained from $f$ by fixing the set of variables \emph{not} in $S$ according to $x$. By construction, $f'$ is defined over $\s(f)$ many variables and is fully sensitive at the input $x|_S$ given by $x$ restricted to the indices in $S$. 
Since $f'$ is a subfunction of $f$ and $g$ is non-constant, we have 
$\adeg(f\circ g) \geq \adeg (f'\circ g)$. 

Notice that $f'$ at the neighbourhood of $x$, in the Boolean cube, is the partial function 
$\PrOR$ (Definition~\ref{defi: promised OR}) or its negation. Therefore, we have $\adeg(f\circ g) \geq \adeg (f'\circ g) \geq \bdeg(\PrOR_{|S|}\circ g)$ (see Definition~\ref{defi: bounded degree} for a definition of the bounded approximate degree).
We can now invoke the composition theorem for $\PrOR$ (Theorem~\ref{thm:promise-or-composition}) \cite{BBG+18} to obtain our lower bound:
\[\adeg(f\circ g) \geq \adeg (f'\circ g) \geq \bdeg(\PrOR_{|S|}\circ g) = \widetilde\Omega(\sqrt{\s(f)}\,\adeg(g)).\] 

However, there is a technical issue with our argument above. When we claimed that $f'$ looks like a $\PrOR$ function we were not quite correct. Technically, it is a Shifted-$\PrOR$ function $\PrOR_{|S|}^{x|_{S}}$, where $\PrOR_n^a(y_1,y_2,\ldots ,y_n) := \PrOR_n(y_1\oplus a_1, y_2\oplus a_2, \ldots , y_n\oplus a_n)$ for $a \in \zone^n$. Formally, we have 
\begin{align}
\label{eq:sens-proof}
\adeg(f\circ g) \geq \adeg (f'\circ g) \geq \bdeg(\PrOR_{|S|}^{x|_{S}}\circ g) = \bdeg(\PrOR_{|S|}\circ (g_1,\ldots ,g_{|S|})),
\end{align}
where $g_i=g$ or $\neg g$ depending on the corresponding $i$-th bit in $x|_{S}$.

We, therefore, need a composition theorem for $\PrOR$ with \emph{different} inner functions, while Theorem~\ref{thm:promise-or-composition} requires that all the inner functions be same. In fact, we would need a more general composition theorem with \emph{different} inner \emph{partial} functions, which we restate below. This generalization is crucially used when dealing with block sensitivity. 
\ApproxDegPrORThm*
The proof of Theorem~\ref{thm: for different function in PrOR} is a generalization of proof of  Theorem~\ref{thm:promise-or-composition}. For the sake of completeness and reader's convenience, we present the proof in Appendix~\ref{appendix: shifted OR} (Theorem~\ref{lem:pror-weak-bounds})\footnote{For a nearly optimal generalization see Theorem~\ref{thm:pror-optimal-bounds} in Appendix~\ref{appendix: shifted OR}.}.     

Now returning to Eq.~\eqref{eq:sens-proof} and using Theorem~\ref{thm: for different function in PrOR}, we obtain the desired lower bound: 
\[\adeg(f\circ g) \geq \adeg (f'\circ g) \geq \bdeg(\PrOR_{|S|}^{x|_{S}}\circ g) = \widetilde\Omega(\sqrt{\s(f)}\,\adeg(g)).\]
We are now ready to present the modifications required to improve the lower bound to $\widetilde\Omega(\sqrt{\bs(f)}\,\adeg(g))$.
\begin{proof}[Proof of Theorem~\ref{thm: blocksensitivity composition}]
Let $b = \bs(f)$ and $a=(a_1,a_2,\ldots ,a_n)$ be an input where $f$ achieves the maximum block sensitivity. Further, let 
$B_1, B_2,\ldots ,B_b$ be disjoint minimal sets of variables that achieves the block sensitivity at $a$, i.e., $f(a) \neq f(a^{B_i})$ for all $i \in [b]$. Recall, $a^{B_i}$ denotes the Boolean string obtained 
from $a$ by flipping the bits at all the indices given by $B_i$. Define a partial function $f': \zone^n \to \{0,1,*\}$ such that, 
\begin{align*}
  f'(x) = \begin{cases}
   0 & \text{if } x=a,\\
   1 & \text{if } x=a^{B_i}, \text{ for some }i \in [b],\\
   * & \text{otherwise}.
  \end{cases}
\end{align*}
Note that $f$ contains $f'$ or its negation as a sub function. 
Thus,  
 \(   \adeg(f \circ g)  \geq \bdeg(f' \circ g) \). 

Since $g$ is non-constant, we can fix the indices \emph{not} in  $\bigcup_{i=1}^bB_i$ according to $a$ to obtain $f''\circ g$. 
We would now like to embed $\PrOR_b$ over the remaining variables in $f''$. For this purpose we define the following partial functions: for every $i\in [b]$, let $I_i:\zone^{B_i}\to \{0,1,*\}$ 
be such that 
\begin{align*}
  I_i(x) = \begin{cases}
   0 & \text{if } x=a|_{B_i}, \\
   1 & \text{if } x=a^{B_i}|_{B_i},\\
   * & \text{otherwise}.
  \end{cases}
\end{align*}
Now observe that $f''\circ g$ can be rewritten as $\PrOR_b\circ(I_1\circ g,\ldots,I_b\circ g)$.  We therefore have   
\begin{align*}
    \adeg(f \circ g) \geq \bdeg(f' \circ g) \geq & ~\bdeg(f'' \circ g) = \bdeg(\PrOR_b\circ(I_1\circ g,\ldots,I_b\circ g)) \\
    & = \Omega\left(\frac{\sqrt{b}\cdot \min_i \bdeg(I_i\circ g)}{\log b} \right) 
     = \widetilde\Omega\left( \sqrt{b}\cdot \adeg(g)\right), 
\end{align*}
where the second last equality follows from Theorem~\ref{thm: for different function in PrOR} and the last equality uses the fact \(\bdeg(I_i\circ g) \geq \adeg(g)\) for all $i$, which in turn follows from each $I_i$ being non-constant.   
\qedhere
\end{proof}

We end this section with few final remarks. 
As a corollary to Theorem~\ref{thm: blocksensitivity composition} we have the following composition for $\adeg$ when the outer function has minimal approximate degree with respect to its block sensitivity. 
\ApproxDegBSCorro*

We also note that the set of Boolean functions with 
$\adeg(f)=\Theta(\sqrt{\bs(f)})$ includes examples of \emph{non-symmetric} functions $f$ with \emph{low} approximate degree. In other words, when such functions are outer function in a composed function then the composition of $\adeg$ doesn't follow from the known results \cite{BBG+18,Sherstov12}. For example, consider the Rubinstein function $\rub$ with arity $n$ (Definition~\ref{defi: Rubinstein function}). It is a non-symmetric function with $\adeg(\rub) = O(\sqrt{n}\log n)$ (Lemma~\ref{claim: ub adeg sink and rub}). Thus, the composition of $\adeg(\rub\circ g)$ does not follow from \cite{BBG+18} or \cite{Sherstov12}. However, $\adeg(\rub) = \tilde\Theta(\sqrt{\bs(\rub)})$, and hence from Corollary~\ref{cor:adeg-minimal-bs}, $\adeg(\rub\circ g) = \tilde\Theta(\adeg(\rub)\, \adeg(g))$. 

Another example is given by the $\sink$ function on $\Theta(n^2)$ variables (Definition~\ref{defi: sink}). 
Its approximate degree is $O(\sqrt{n}\log n)$ (Lemma~\ref{claim: ub adeg sink and rub}), while $\bs(\sink) = \Theta(n)$. Thus, $\adeg(\sink\circ g) = \tilde\Theta(\adeg(\sink)\,\adeg(g))$ follows from Corollary~\ref{cor:adeg-minimal-bs}. 

As stated in the introduction, we recall that Theorem~\ref{thm: blocksensitivity composition} is tight in terms of block-sensitivity, i.e., the lower bound can not be improved to $\widetilde\Omega(\bs(f)^c\cdot\adeg(g))$ for some $c > 1/2$.


\section{Conclusion}
While our work makes progress on the composition problem for $\R$ and $\adeg$, the main problems of whether $\adeg$ and $\R$ composes for any pair of Boolean functions still remains open. In this light, we would like to highlight some questions that can be useful stepping stones towards the main questions. 



We showed that the composition question for $\R$ is equivalent to the following open question (which is a generalization of Ben-David and Blais~\cite{bb20focs} result):


\begin{open question}
 Let $f:\zone^n \rightarrow \zonep$ be a Boolean function. 
 Then, 
 is it true that for arbitrary $t$, $\noisyR(f) = \Theta\left(\R(f \circ \GapMaj_t) /t \right)$?
 \end{open question}



In case of approximate degree composition, 
a natural question is whether $\sqrt{\bs(f)}$ can be replaced by some other complexity measures. In this regards we state the following open problems: 


\begin{open question}
For all Boolean functions $f$ and $g$, can we prove either of the following: 
\begin{minipage}[t]{0.5\textwidth}
\begin{itemize}
\item $\adeg( f\circ g) = \Omega(\sqrt{ \deg(f)} \cdot \adeg(g))$?
\end{itemize}
\end{minipage}
\begin{minipage}[t]{0.5\textwidth}
\begin{itemize}
\item $\adeg( f\circ g) = \Omega( \sqrt{ \fbs(f)} \cdot \adeg(g))$? 
\end{itemize}
\end{minipage}
\end{open question}
Recently, in \cite{SYZ04, Sun07, Drucker11, Chakraborty11, DBLP:conf/icalp/0001KP22}, the classes of transitive functions got a lot of attention as natural generalization of the classes of symmetric functions. Can the result for symmetric functions be extended to transitive functions? 

\begin{open question}
Can we prove that $\adeg$ and $\R$ compose when the outer function is transitive?
\end{open question}

\bibliographystyle{alpha}
\bibliography{reference}
\appendix

\section{Some important Boolean functions}
In this section, we define some important Boolean functions that have been used in the paper. We start with Multiplexer Function or Addressing Function.

\begin{defi}[Multiplexer Function or Addressing Function]
\label{def: mux or addresssing}
The function $\mathsf{MUX}:\zone^{k+2^k}\to\zone$ with input $(x_0,\dots, x_{k-1}, y_0, \dots, y_{2^k-1})$ outputs the bit $y_t$, where $t = \sum_{i=0}^{k-1} x_i 2^{i}$.
\end{defi}

The following function was defined in~\cite{CMS20}.

\begin{defi}[$\sink$]
\label{defi: sink}
Consider a tournament defined on $k$ vertices with $\binom{k}{2}$ variables such that, for $i < j$, if $x_{ij} = 1$ then there is an outgoing edge from $i$ to $j$.
A vertex $i \in [n]$ is a sink-vertex if all edges incident to it are incoming edges. For $x \in \binom{k}{2}$, $\sink(x) = 1$ if there a vertex $i \in [n]$ such that $i$ is a sink-vertex, and $0$ otherwise. 
\end{defi}

\begin{defi}[Rubinstein function~\cite{Rub95}]
\label{defi: Rubinstein function}
Let $g: \zone^k \to \zone$ be such that $g(x) =1$ iff $x$ contains two consecutive ones and the rest of the bits are $0$. The Rubinstein function, denoted by $\rub: \zone^{k^2} \to \zone$, is defined to be $\rub = \OR_k \circ g$.
\end{defi}

\subsection{Some properties of $\sink$ and Rubinstein function}
Following is an easy observation.
\begin{observation}
\label{obs: bs of sink}
The sensitivity of $\sink: \zone^{\binom{k}{2}} \to \zone$ is at least $(k-1)$. Consider a tournament on vertices $v_1, \dots, v_k$ such that $v_1$ is a sink-vertex and $(v_2, \dots, v_k, v_2)$ is a cycle. Observe that flipping any edge incident to $v_1$ changes the value of the function from $1$ to $0$. In particular, \(\bs(\sink) \geq \s(\sink) = \Omega(k)\). 
\end{observation}

We now want to give an upper bound on the approximate degree of these two functions. For this, we first need the following generalization of approximate degree.

\begin{defi}[$\varepsilon$-Error Approximate Degree $(\adeg_\varepsilon)$]
\label{defi: eps approximate degree}
A polynomial $p: \mathbb{R}^n \to \mathbb{R}$ is said to $\varepsilon$-approximate a Boolean function $f: \zone^n \to \{0,1\}$ if 
\begin{align*}
    | p(x) -f(x) | \leq \varepsilon, \quad \forall x \in \zone^n.
\end{align*}
The $\varepsilon$-approximate degree of $f$ $(\adeg_{\varepsilon}(f))$ is the minimum possible degree of a polynomial which $\varepsilon$-approximates $f$.
\end{defi}

It is known (see~\cite{BNRdW07} also~\cite[Appendix A]{tal2014shrinkage}) that given a polynomial $p$ that approximates a Boolean function $f$ to error $1/3$, one can come up with a polynomial $p'$ that $\varepsilon$-approximates $f$ such that $\deg(p') = O(\deg(p) \log (1/\varepsilon))$. 

Also, the following theorem will be used in this section.

\begin{theorem}[\cite{She13b}]
\label{theo: robust poly}
For all Boolean functions $f$ and $g$, $\adeg(f \circ g) = O(\adeg(f) \adeg(g))$.
\end{theorem}

We now upper bound the approximate degree of $\sink$ and $\rub$.

\begin{lemma}
\label{claim: ub adeg sink and rub}
For the sink function $\sink: \zone^{\binom{k}{2}} \to \zone$ 
 and the Rubinstein function $\rub:\zone^{k^2} \to \zone$ we have
 \begin{itemize}
     \item[1.] $\adeg(\sink) = O(\sqrt{k} \log k)$, and 
     \item[2.] $\adeg(\rub) = O(k \log k)$.
 \end{itemize}
\end{lemma}
\begin{proof}
We first upper bound the approximate degree of $\sink$. For every vertex $i \in [n]$, there is a sink at that vertex if and only if all edges incident to that vertex are incoming. This is equivalent to saying that the variables $x_{ij}$, when $j > i$, are $0$ and $x_{ji}$, when $j < i$, are $1$. This can be computed by an $\AND_k$ function and we call this function $\AND_k^{(i)}$. 
Also, since any tournament has at most one sink vertex, $\sink$ can be expressed as the sum of $\AND_k^{(i)}$'s, for $i \in [k]$. Recall that the approximate degree of $\AND_k$ is $O(\sqrt{k})$, which, from Theorem~\ref{theo: robust poly}, implies that $1/(3k)$-error approximate degree of $\AND_k$ is $O(\sqrt{k}\log k)$. Replacing each $\AND_k^{(i)}$ with $1/(3k)$-error approximating polynomial gives a $1/3$-error approximating polynomial for $\sink$ with degree $O(\sqrt{k}\log k)$.

Now we upper bound the approximate degree of Rubinstein function. Recall that $\rub: \zone^{k^2} \to \zone$ is defined as $\rub = \OR_k \circ g$, where $g:\zone^k \to\zone$ is a function such that $g(x) = 1$ if and only if $x$ contains two consecutive $1$'s and the rest of the bits are $0$. Observe that there are only $(k-1)$ inputs on which $g$ takes value $1$, call these inputs $z_1, \dots, z_{k-1}$. Let $\AND_k^{(i)}$ be the Boolean function which evaluates to $1$ if and only if its input is $z_i$. Thus $g$ can be expressed as the sum of $\AND_k^{(i)}$'s, for $i \in [k]$. By a similar argument as in the last paragraph, the $1/3$-error approximate degree of $g$ can be bounded by $O(\sqrt{k}\log k)$. From Theorem~\ref{theo: robust poly}, this implies an $O(k \log  k)$ upper bound on the approximate degree of $\rub$.
\end{proof}

\section{Approximate degree of Promise-OR composed with different inner functions}
\label{appendix: shifted OR}




In this section we show that the approximate degree composes when the outer function is $\PrOR$ and the inner functions are (possibly) different partial functions. The proof is essentially a straightforward generalization of the proof of Theorem~\ref{thm:promise-or-composition} \cite[Theorem 16 (arXiv version)]{BBG+18}. However, for the sake of completeness and reader's convenience, we give an overview of the proof here. We will need some definitions and theorems from \cite{BBG+18} which we state now. We start with the definition of a problem called ``singleton combinatorial group testing''. It generalizes the combinatorial group testing problem.

\begin{defi}[Singleton CGT]
Let $D$ be the set of all $w\in\zone^{2^n}$ for which there exists an $x \in \zone^n$ such that for all $S\subseteq [n]$ satisfying $\sum_{i\in S}x_i \in \zone$, we have $\sum_{i \in S}x_i = w_S$. Note that for all $w\in D$, the string $x$ is uniquely defined by $x_i = w_{\{i\}}$. Let us denote this string by $x(w)$. we then define the partial function 
$\mathsf{SCGT}_{2^n}:D\to\zone^n$ by $\mathsf{SCGT}_{2^n}(w) = x(w)$. 
\end{defi}

\begin{theorem}[{\cite[Theorem~19 (arXiv version)]{BBG+18}}]
\label{thm:quantum-SCGT}
 The bounded-error quantum query complexity of $\mathsf{SCGT}_{2^n}$ is $\Theta(\sqrt{n})$. 
\end{theorem}

For a formal Definition of bounded error quantum query complexity we refer the survey by \cite{BW02}.
Before we state the next result that we need from \cite{BBG+18} we are defining robustness of a polynomial to input noise.

\begin{defi}[Robustness to input noise]
    \label{defi: robustness of a polynomial}
   For any function $f: \zone^n \to \{0,1,*\}$ we say a polynomial $p: \zone^n \to \R$ approximately computes $f$ with $\delta$-robustness where $\delta \in [0, \frac{1}{2})$ if for any $x \in \Dom(f)$ and $\Delta \in [-\delta, \delta]^n$, we have $|f(x) - p(\Delta+ x)| \leq \frac{1}{3}$. 
\end{defi}

Now we are ready to state the next result. 

\begin{theorem}[{\cite[Theorem~17 (arXiv version)]{BBG+18}}]
\label{thm:quantum-robustness}
For a partial Boolean function $f$, there exists a bounded multilinear polynomial $p$ of degree $O(\mathsf{Q}(f))$ that approximates $f$ with robustness $\Omega(1/\mathsf{Q}(f)^2)$ where $\mathsf{Q}(f)$ is the bounded error quantum query complexity of the function $f$. 
\end{theorem}

We refer \cite{BBG+18} for more details about robustness of a polynomial induces by quantum algorithm.
We also need the existence of a multilinear robust polynomial for $\mathsf{XOR}_n\circ\mathsf{SCGT}_{2^n}$, which follows from Theorems~\ref{thm:quantum-SCGT} and \ref{thm:quantum-robustness} above, where $\mathsf{XOR}_n\circ\mathsf{SCGT}_{2^n}$ is the parity of $n$ output bits of $\mathsf{SCGT}_{2^n}$. 

\begin{theorem}[{\cite[Theorem~20 (arXiv version)]{BBG+18}}]
\label{thm:robust-poly-SCGT}
There is a real polynomial $p$ of degree $O(\sqrt{n})$ over $2^n$ variables $\{w_S\}_{S\subseteq [n]}$ and a constant $c \geq 10^{-5}$ such that for any input $w \in \zone^{2^n}$ with $\mathsf{XOR}_n\circ\mathsf{SCGT}_{2^n}(w) \neq *$ and any $\Delta \in [-c/n,c/n]^{2^n}$, 
\[|p(w + \Delta) - \mathsf{XOR}_n\circ\mathsf{SCGT}_{2^n}(w)|\leq 1/3.\]
Furthermore, $p$ is multilinear and for all $w\in\zone^{2^n}$, $p(w) \in [0,1]$.
\end{theorem}
We also need the following result of Sherstov that shows composition holds for the approximate degree of the parity of $n$ different functions.
\begin{theorem}[{\cite[Theorem~5.9]{Sherstov12}}]
\label{thm:sherstov-xor}
For any partial Boolean functions $f_1,\ldots ,f_n$, we have 
\[\bdeg(\mathsf{XOR}\circ(f_1,\ldots ,f_n)) = \Omega\left(\sum_{i=1}^n\bdeg(f_i)\right).\]
\end{theorem}

We are now ready to prove Theorem~\ref{thm: for different function in PrOR} which we restate below.

\begin{theorem}
\label{lem:pror-weak-bounds}
For any partial Boolean functions $f_1, f_2, \ldots, f_n$, we have

\[\bdeg\left(\PrOR_n\circ (f_1,f_2,\ldots ,f_n)\right) = \Omega\left(\frac{\sqrt{n}\cdot \min_{i=1}^n\bdeg(f_i)}{\log n}\right).\]

Furthermore the following upper bound also holds, 
\[\bdeg\left(\PrOR_n\circ (f_1,f_2,\ldots ,f_n)\right) = O\left(\sqrt{n}\cdot \max_{i=1}^n\bdeg(f_i)\cdot \log n\right).\] 

\end{theorem}
\begin{proof}
The upper bound follows by first amplifying the approximation of inner function to within error $\Theta(1/n)$ and then composing with the polynomial given by Theorem~\ref{thm:quantum-robustness} for \text{PrOR}.

For amplification one can use the univariate amplification polynomial of degree $O(\log(1/\varepsilon))$ that maps $[0,1/3]$ to $[0,\varepsilon]$, $[2/3,1]$ to $[1-\varepsilon,1]$, and $[1/3,2/3]$ to $[0,1]$ given by  \cite[Lemma~1]{BNRdW07}. 

We now give an overview of the lower bound proof following \cite{BBG+18}. 

Let $q$ be an approximating polynomial for $\PrOR\circ(f_1,\ldots ,f_n)$ of degree $T:=\bdeg(\PrOR\circ(f_1,\ldots ,f_n))$, which is also bounded on all Boolean inputs outside the promise.  Let $q'$ be the polynomial obtained from $q$ by amplifying the approximation to within error $c/n$ on all inputs in the promise, where $c$ is the constant from Theorem~\ref{thm:robust-poly-SCGT}. Note that the degree of $q'$ is $O(T\cdot\log n)$, and it remains bounded on all possible Boolean inputs.

We can assume without loss of generality that $f_1,f_2,\ldots ,f_n$ are non-constant partial Boolean functions. Therefore for each $f_i$ there exists an input $y$ such that $f_i(y) =0$. 
For all $S\subseteq [n]$, we now define a polynomial $q'_S$ using $q$ that approximates $\PrOR\circ ( f_i)_{i\in S}$ by setting the variables of each $f_i$, $i\not\in S$, to an input where it evaluates to $0$. Note that the degree of $q'_S$ is bounded by the degree of $q'$.  

Now consider the polynomial $p$ over $2^n$ variables $\{w_S\}_{S\subseteq[n]}$ given by Theorem~\ref{thm:robust-poly-SCGT}. 
Let $r$ be the polynomial obtained from $p$ by replacing the variables $w_S$ by polynomials $q'_S$, i.e., $r = p\circ (q'_S)_{S\subseteq[n]}$. 
Clearly the degree of $r$ is $O(T\sqrt{n}\log n)$. 

It can now be argued that $r$ approximates $\mathsf{XOR}\circ(f_1,\ldots ,f_n)$ to error within $1/3$. A slight care is needed in this argument for $p$ expects inputs  which are $\Delta$-close to Boolean values $\zone$. However, it may happen that some $q'_S$, though bounded in $[0,1]$, is not close to Boolean values $\zone$. This is where we will use the fact that $p$ is also \emph{multilinear}, and hence the value of $p$ on a convex combination of Boolean inputs is equal to the convex combination of values of $p$ on the Boolean inputs. Hence, if $p$ works correctly when all inputs are in $\zone$, then it must also be correct on inputs in $[0,1]$. A final thing to note is that any invalid input (to $p$) in $[0,1]^{2^n}$ can be written as a convex combination of valid inputs.

Thus, we have
\[\bdeg(\mathsf{XOR}\circ(f_1,\ldots ,f_n)) = O(T\sqrt{n}\log n) = O(\sqrt{n}\log n\cdot \bdeg(\PrOR\circ(f_1,\ldots ,f_n))).\]
Whereas from Theorem~\ref{thm:sherstov-xor} we have
\[\bdeg(\mathsf{XOR}\circ(f_1,\ldots ,f_n)) = \Omega\left(\sum_{i=1}^n\bdeg(f_i)\right) = \Omega\left( n \cdot \min_i \bdeg(f_i)\right).\]
Combining the two, we obtain the lower bound
\[\bdeg(\PrOR\circ(f_1,\ldots ,f_n)) = \Omega\left(\frac{\sqrt{n}\cdot \min_{i=1}^n\bdeg(f_i)}{\log n}\right).\]

\end{proof}
We will now use this weak bound to establish nearly optimal bound for the approximate degree of \text{PrOR} composed with $n$ different \emph{partial} functions. This will again be a simple generalization of \text{OR} composed with different functions \cite[Theorem~37]{BBG+18}. For the sake of completeness, we work out some of the details.


\begin{theorem}
\label{thm:pror-optimal-bounds}
For any partial Boolean functions $f_1, f_2, \ldots, f_n$, we have
\begin{align*}
    \bdeg\left(\PrOR_n \circ (f_1,f_2,\ldots, f_n)\right) = \widetilde\Theta\left(\sqrt{\sum_{i=1}^n\bdeg(f_i)^2}\right),
\end{align*}
when the lcm of $\bdeg(f_i)^2$ for $i\in[n]$ is $\Theta(\max_i~ \bdeg(f_i)^2)$.
\end{theorem}
\begin{proof}
As mentioned before, the proof is merely working out the details of \cite[Theorem~37]{BBG+18} while keeping in mind that we are working with \emph{partial} functions.  

Let $F = \PrOR_n \circ (f_1,f_2,\ldots, f_n)$, $d_i = \bdeg(f_i)^2$ for $i\in[n]$, and $\ell$ be the lcm of $d_i$'s. Now consider the function $G= \PrOR_{\ell}\circ F$. From Theorem~\ref{lem:pror-weak-bounds}, we have the following bounds on $\bdeg(G)$ up to constants
\begin{align}
\label{eq:bdeg-G-wrt-F}
\frac{\sqrt{\ell}\cdot \bdeg(F)}{\log \ell} \leq \bdeg(G) \leq \sqrt{\ell}\cdot\bdeg(F)\cdot\log \ell.
\end{align}
Now using the associativity of $\PrOR$ we can rewrite $G$ as 
\begin{align}
    G = \PrOR_{n\ell}\circ(\underbrace{f_1,\ldots ,f_1}_{\ell \text{ times}},\ldots ,\underbrace{f_n,\ldots ,f_n}_{\ell \text{ times}} ). 
\end{align}
Further regrouping $f_i$'s, we can rewrite $G$ as follows
\begin{align}
\label{eq:regrouping-G}
    G = \PrOR_d\circ(\underbrace{\PrOR_{\ell/d_1}\circ f_1,\ldots ,\PrOR_{\ell/d_1}\circ f_1}_{d_1 \text{ times}}, \ldots ,\underbrace{\PrOR_{\ell/d_n}\circ f_n,\ldots ,\PrOR_{\ell/d_n}\circ f_n}_{d_n \text{ times}}), 
\end{align}
where $d = \sum_{i=1}^n d_i$. Now using Theorem~\ref{lem:pror-weak-bounds} and $\sqrt{d_i}=\bdeg(f_i)$, we obtain following bounds for $\PrOR_{\ell/d_i}\circ f_i$ (up to constants)
\begin{align}
\label{eq:bdeg-regroup}
    \frac{\sqrt{\ell}}{\log (\ell/d_i)} \leq \bdeg(\PrOR_{\ell/d_i}\circ f_i) \leq \sqrt{\ell}\cdot \log(\ell/d_i).
\end{align}
Now consider \eqref{eq:regrouping-G} and using Theorem~\ref{lem:pror-weak-bounds} along with \eqref{eq:bdeg-regroup}, we obtain
\begin{align}
\label{eq:bdeg-G}
    \frac{\sqrt{d\ell}}{\log d \cdot \log \ell}\leq \bdeg(G) \leq \sqrt{d\ell}\cdot \log \ell \cdot \log d
\end{align}
Now from \eqref{eq:bdeg-G} and \eqref{eq:bdeg-G-wrt-F} it follows 
\[
\frac{\sqrt{d}}{\log d \cdot \log^2 \ell}\leq \bdeg(F) \leq \sqrt{d}\cdot \log^2 \ell \cdot \log d .
\]
\end{proof}

\section{Composition theorems for strongly-$k$-junta symmetric outer functions}
\label{appendix: junta symmetric function}

In this section we will prove the composition result of $\adeg$ and $\R$ when the outer function has some amount of symmetry. Of course, there are various notion of symmetry. Traditionally a function is said to have the maximum amount of symmetry when the function value is invariant under any permutation of the variables. Such functions are called symmetric.  Symmetric functions are very well studied in the literature of Boolean function analysis. In the terms of composition theorems of $\adeg$ and $\R$ it was proved in \cite{BBG+18} and \cite{GJPW18} that $\adeg$ and $\R$ respectively composes when the outer function is symmetric. 

In terms of weaker notions of symmetry there are various possible definitions. In this paper we consider the case of strongly-$k$-junta symmetric functions. The composition theorem for $\adeg$ when the outer function is strongly-$k$-junta symmetric (Theorem~\ref{thm: junta sym composition Intro}(Part(i)) is presented in Appendix~\ref{sec:juntaadeg}. The proof of the composition theorem for $\R$ when the outer function is strongly-$k$-junta symmetric (Theorem~\ref{thm: junta sym composition Intro}(Part(ii)) follows easily from Theorem~\ref{theo: composition with full RIntro}.



\begin{observation}
\label{thm: junta sym composition for R}
For any strongly $k-$junta symmetric function $f: \zone^n \to \zone$ and any Boolean function $g: \zone^m \to \zone$, we have $\rqc(f\circ g) = \widetilde{\Omega}({\rqc(f)\cdot \rqc(g)})$ where $ n-k = \Theta(n)$.
\end{observation}
\begin{proof}
There exists an assignment of the $k$-bits such that the resulting function is a non-constant symmetric function on $(n-k)$ bits. Since the sensitivity of the restricted function is $\Omega(n)$, the randomized query complexity is also $\Omega(n)$ (see~\cite{nisan1989crew}).
Hence, from Theorem~\ref{theo: composition with full RIntro} the result follows.
\end{proof}


\subsection{Composition of approximate degree for $\sqrt{n}$-junta symmetric functions}\label{sec:juntaadeg}

A crucial result that we use in the prove of composition theorem of $\adeg$
is the following result from \cite{Paturi}.
\begin{theorem}[\cite{Paturi}]
\label{thm:paturi-sym}
For any non-constant symmetric function $f:\zone^n\to\zone$, let $k$ be the closest integer to $n/2$ such that $f$ takes different values on inputs of Hamming weight $k$ and $k+1$. Define, 
\[\gamma(f) = 
\begin{cases}
k & \text{if } k\leq n/2,\\
n-k & \text{otherwise}.
\end{cases}
\]

Then
\[\adeg(f) = \Theta\left(\sqrt{n(\gamma(f)+1)}\right).\]
\end{theorem}

Using the result of \cite{Paturi} we prove the following proposition about the approximate degree of a $k$-junta symmetric function. Recall the multiplexer function from Definition~\ref{def: mux or addresssing}.


\begin{prop}
\label{prop:adeg-k-junta}
For any $k$-junta symmetric function $f: \zone^n \to \zone$, we have 
$\adeg(f)=\Omega\left(\sqrt{(n-k)\gamma_{\max}}\right)$ and $\adeg(f)=O\left(\max\{k,\sqrt{(n-k)\gamma_{\max}}\}\right)$, 
where $\gamma_{\max} = \max_{i\in\zone^k}\{\gamma(f_i)\}$ 
such that $f_i$ is the symmetric function obtained by restricting 
the junta variables according to $i$. 
\end{prop}
\begin{proof}
Fixing the junta variables in $f$ we obtain a symmetric function on $n-k$ variables with approximate degree $\Omega(\sqrt{(n-k)\gamma_{\max}})$ (Theorem~\ref{thm:paturi-sym}), which in turn implies the same lower bound on $\adeg(f)$.   

For the upper bound, we obtain an approximating polynomial for $f$ 
by composing the (exact) polynomial for the multiplexer function 
$\mathsf{MUX}:\zone^{k+2^k}\to\zone$ with the approximating polynomials for different symmetric functions obtained by restricting the $k$ junta variables. 
Therefore, $\adeg(f) = k + O(\sqrt{(n-k)\gamma_{\max}}) = O\left(\max\{k,\sqrt{(n-k)\gamma_{\max}}\}\right)$. 

\end{proof}

As mentioned earlier, the composition of $\adeg$ when the outer function is symmetric was proved in \cite{BBG+18}.  The following is their result that we crucially  use in the proof of Theorem~\ref{thm: junta sym composition Intro}.

\begin{theorem}[\cite{BBG+18}]
\label{thm:sym-composition}
For any symmetric Boolean function $f: \zone^n \to \zone$ and any Boolean function $g: \zone^m \to \zone$ we have, 
\begin{align*}
    \adeg(f \circ g) = \Tilde{\Omega}(\adeg(f) \cdot \adeg(g))
\end{align*}
\end{theorem}

We now present the proof of Part (i) of Theorem~\ref{thm: junta sym composition Intro}, that the proof of composition of $\adeg$ when the outer function is strongly-$k$-junta symmetric.


\begin{proof}[Proof of Theorem~\ref{thm: junta sym composition Intro}(Part (i))]
Since $f$ is a strongly-$k$-junta symmetric function so there exists a setting of the $k$ junta variables such that the resulting function is a non-constant symmetric function. 
Let $f'$ be the symmetric function obtained by 
restricting the junta variables of $f$ so that $f'$ is non-constant. 
Then by Theorem~\ref{thm:paturi-sym} the approximate degree of $f'$ is $\Omega(\sqrt{(n-k)\gamma_{\max}})$. Then clearly we have 
\begin{align}
\label{eq:junta-symm}
    \adeg(f\circ g) \geq \adeg(f'\circ g) = 
    \widetilde{\Omega}(\adeg(f')\cdot \adeg(g))  = \widetilde{\Omega}(\sqrt{(n-k)\gamma_{\max}}\cdot \adeg(g)),
\end{align}
where the first equality follows from Theorem~\ref{thm:sym-composition}. 
Now from Proposition~\ref{prop:adeg-k-junta} we know that $\adeg(f) = O(\sqrt{(n-k)\gamma_{\max}})$ if $k = O(\sqrt{(n-k)\gamma_{\max}})$, which is satisfied when $k=O(\sqrt{n})$. Thus from \eqref{eq:junta-symm} we obtain 
\[\adeg(f\circ g) = \widetilde{\Omega}(\adeg(f)\cdot \adeg(g)).\]

\end{proof}



\section{Proof of Theorem~\ref{thm: bdb20 main theorem Intro}} 
\label{section: proof of BB20 main thm}
Let us start by recalling the theorem we want to prove in this section.
\GeneralizedCharNoisyR*


We mention two lemmas from Ben-David and Blais~\cite{BB20} that we need for the proof of this theorem. The first lemma shows that given a noisy oracle with $\gamma$ bias, it can simulate an oracle of $\gamma' \geq \gamma$ bias by making approximately $\left(\frac{\gamma'}{\gamma}\right)^2$ queries.

\begin{lemma}[\cite{BB20}]
\label{lemma: Bias Amplification Lemma}


Let $\gamma \in [ - 1/3, 1/3]$ be nonzero and $k \leq 1/\gamma^2$ be an odd positive integer. If we take $k$ independent samples from $\text{Bernoulli}((1+\gamma)/2)$ and take their majority, the resultant distribution is $\text{Bernoulli}((1+\gamma')/2)$ where $\sign(\gamma') = \sign(\gamma)$ and
\begin{align*}
    \frac{\sqrt{k} |\gamma|}{3} \cdot \leq |\gamma'| \leq 3 \sqrt{k} |\gamma|.
\end{align*}
\end{lemma}

The following lemma shows that it is enough to query the noisy oracle with just two biases.

\begin{lemma}[\cite{BB20}]
\label{lemma: noisyR only two bias}
Let $f$ be a partial function and let $A$ be an optimal noisy randomized algorithm for $A$ of cost $\noisyR(f)$. Then there is another noisy randomized algorithm $\widehat{A}$ of cost $O(\noisyR(f))$, which only queries its noisy
oracles with parameter either $\gamma = 1$ or $\gamma = \widehat{\gamma}$, where $\widehat{\gamma} > 0$ is the smallest bias that $A$ uses on any input.
\end{lemma}

Another important lemma needed for the proof concerns the property of a random walk on a line where the coin is biased with some probability.

\begin{lemma}
\label{lemma: random walk with bias}
Define a random walk on a line where the coin gives head with probability $\frac{1+\widehat{\gamma}}{2}$. The random walk start at $0$ and stops if it reaches $T$ or $-T$. Conditioned on the fact that we reach $T$ before $-T$, let the expected number of steps of the walk be $\mu_T$, then 

\begin{align}
    \mu_{T}
    &= \frac{T}{\widehat{\gamma}}\ -\ \frac{2T}{\widehat{\gamma}}
    (1 - \widehat{\gamma})^{T}
    \left(\frac{(1 + \widehat{\gamma})^{T} - (1 - \widehat{\gamma})^{T}}{(1 + \widehat{\gamma})^{2T} - (1 - \widehat{\gamma})^{2T}}\right). \tag{\cite[Chapter XIV]{feller}} \label{eqn: mu T}
\end{align}

Moreover, for $T = \Theta(1/{(\sqrt{t} \cdot \widehat{\gamma})})$, we have
\begin{itemize}
    \item $\mu_{2T} = \Omega(1/{(t \cdot \widehat{\gamma}^2}))$, and
    
    \item $\mu_{2T} \leq 12 \mu_T$.
\end{itemize}
\end{lemma}

For completeness, we present the proof after the proof of  Theorem~\ref{thm: bdb20 main theorem Intro}

\begin{proof}[Proof of  of Theorem~\ref{thm: bdb20 main theorem Intro}]
From Lemma~\ref{lemma: noisyR only two bias}, we can assume that $f$ can be computed by a noisy randomized algorithm $\widehat{A}$ of cost $O(\noisyR(f))$ that makes queries with only two biases: $1$ and ${\widehat{\gamma}}$. We will now simulate $\widehat{A}$, that uses bias $1$ and $\widehat{\gamma}$, with a $\noisyR$ algorithm $B$ that uses bias $1$ and $1/\sqrt{t}$, where $t \geq 1$.

\ 

\textbf{Case 1:} The first case is when $\widehat{\gamma} \geq 1/\sqrt{t}$. In this case we can use multiple oracle calls of bias $1/\sqrt{t}$ to simulate one oracle call of bias $\widehat{\gamma}$. To do this, algorithm $B$ makes $O(t \widehat{\gamma}^2)$ calls of bias $\widehat{\gamma}$ and takes their majority. By Lemma~\ref{lemma: Bias Amplification Lemma}, the algorithm $B$ obtains a bit of bias slightly greater than $\widehat{\gamma}$. Then $B$ adds some bias to this bit to obtain a bit of bias exactly $\widehat{\gamma}$. The cost paid by $B$ is $O(t \widehat{\gamma}^2 \cdot 1/t) = O(\widehat{\gamma}^2)$.

\ 

\textbf{Case 2:} We now consider the case when $\widehat{\gamma} < 1/\sqrt{t}$. In this case we wish to generate many bits of low bias ($\widehat{\gamma}$) using a single bit of high bias ($1/\sqrt{t}$). Consider the following random walk based sampling procedure.
    
    
    

\begin{itemize}
    \item \textbf{Setup.}
    Let 
    \begin{align}
       T = \left\lfloor \frac{1}{5\sqrt{t}\widehat{\gamma}} \right\rfloor. \label{eqn: setting T} 
    \end{align}
    The random walk will take place on a line that is marked with integral multiples of $T$, with $T$ being right to $0$, $-T$ being left to $0$ and so on. Also, $(T-1)$ points are marked between every two adjacent integral multiples of $T$. The random walk starts at $0$. Also, let
    \begin{align*}
        R = \left(\frac{1+\widehat{\gamma}}{1-\widehat{\gamma}}\right)^T.
    \end{align*}
    
    \item[1.] Toss a coin of bias 
    $$
    \delta' = \frac{R-1}{R+1}. 
    $$
    \item[2.] If the result of the toss is $1$, then sample a sequence of bits using a $\widehat{\gamma}$-bias coin consisting of $w$ $1$'s and $z$ $0$'s, conditioned on  $w - z = T$. 
    This sampling procedure can be simulated by picking a random walk from the set of $\widehat{\gamma}$-biased random walks on the line which reach $T$ before $-T$ (starting at $0$).

    \item[3.] If the result of the toss is $0$, then sample a sequence of bits using a $\widehat{\gamma}$-bias coin consisting of $w$ $1$'s and $z$ $0$'s, conditioned on  $w - z = -T$. 
    This sampling procedure can be simulated by picking a random walk from the set of $\widehat{\gamma}$-biased random walks on the line which reach $-T$ before $T$ (starting at $0$).
    
\end{itemize}

We now prove the correctness of the above protocol. A crucial observation in~\cite{BB20} is that conditioned on a set of sequences, all of which reach $T$ before $-T$, the probability of choosing each sequence is the same whether we choose bias $\widehat{\gamma}$ or $-\widehat{\gamma}$ for Step 2. and Step 3. above. The reason is that the probability for any sequence of walks with $w$ $1$'s and $z$ $0$'s such that $w - z = T$, is $R$ times more likely when we choose a $\widehat{\gamma}$ biased coin than when we choose a $-\widehat{\gamma}$ biased coin. Furthermore, the probability of a $\widehat{\gamma}$-biased walk starting at $0$ and reaching $T$ is $R/(R+1)$. This is because if the probability of reaching $T$ is $p$, then probability of reaching $-T$ is $p/R$. Since $p + p/R = 1$, $p = R/(R+1)$. Similarly, the probability that a $-\widehat{\gamma}$-biased coin reaches $T$ is $1/(R+1)$ and reaches $-T$ is $R/(R+1)$. Thus $\delta'$ is chosen such that the probability of being correct is exactly $R/(1+R)$. Also, assuming that $\widehat{\gamma}$ is smaller than $1/10$, the bias $\delta'$ can be shown to be smaller than but within a constant factor of $1/\sqrt{t}$ (we refer to~~\cite[Theorem~4]{BB20} for details). Thus, by adding noise, the bias $1/\sqrt{t}$ can be converted to bias $\delta'$.

Let $\mu_T$ denote the expected number of bits that one run of the above sampling procedure generates. Lemma~\ref{lemma: random walk with bias} implies that $\mu_T$ is lower bounded by $\Omega(1/(t \cdot \widehat{\gamma}^2))$.

In summary, for one iteration of the Step 1.~to 3. of the sampling procedure, the following are true: a) The choice of $\delta'$ in Step 1.~implies correct distribution (from bias $\widehat{\gamma}$) of generating the coin tosses, and b) an expected number of $\Omega(1/(t\widehat{\gamma}^2))$ $\widehat{\gamma}$-biased samples are generated. c) Cost paid is $\delta^2$.

Before moving further let us mention explicitly how algorithm $B$ works: $B$ runs algorithm $\widehat{A}$ on the given input. For each $i \in [n]$:
\begin{itemize}
    \item If $\widehat{A}$ makes a query of bias $1$ then $B$ also queries that bit with bias $1$. The value of such bits are known with certainty and no other queries to these bits are made.
    
    \item When $\widehat{A}$ makes a query of bias $\widehat{\gamma}$ then 
    \begin{itemize}
        \item If $B$ has $\widehat{\gamma}$ biased bits available, then $B$ uses these bits to run $\widehat{A}$. Otherwise $B$ generates the above sampling procedure to generate $\widehat{\gamma}$ biased bits.
    \end{itemize}
\end{itemize}

Clearly, the correctness of $B$ follows from the correctness of $\widehat{A}$. In order to upper bound the expected cost of $B$, we upper bound the expected number of times $B$ runs the above sampling procedure for each index $i \in [n]$.
Recall that our goal is to show that the total expected cost is bounded by:
\[
    O\left( \noisyR(f) + n/t \right).
\]

Fix a $k \in [n]$ and consider the noisy query algorithm $B$ that uses bias $1/\sqrt{t}$ and $1$ to simulate the algorithm $\widehat{A}$ (that uses bias $\widehat{\gamma}$ and $1$).
For a run, say $r$, of algorithm $\widehat{A}$
\begin{itemize}
    \item Let $T_{r,k}$ be the number of times algorithm $\widehat{A}$ queries the $k$th bit in the $r$-th run and let
    \[
        T_k = \EE_r[T_{r,k}].
    \]
    Observe that
    \begin{align}
        \sum_k T_k = O(\noisyR(f)/\widehat{\gamma}^2). \label{eqn: Tk sums to noisyR}
    \end{align}
    
    \item Let $X_1, X_2, \dots, X_{\ell_r}$ be the number of bits sampled by $B$, while simulating the $r$-th run of $\widehat{A}$, where $B$ makes $\ell_r$ independent calls to the above sampling procedure in order to meet the demand of $T_{r,k}$.
\end{itemize}

Let
\begin{align}
    L_{k,r} = (X_1 + \dots + X_{\ell_r}) - T_{r,k},
\end{align}
denote the number of surplus bits (i.e. not used in simulation) generated by the algorithm $B$.
Thus,
\begin{align}
    \EE[X_1 + \dots + X_{\ell_r}] 
    &= \EE[T_{r,k}] + \EE[L_{k,r}] \nonumber \\
    &= T_k + \EE[L_{k,r}], \label{eqn: main ub.}
\end{align}
where the expectation in the first equation is over both the run $r$ and the randomness of the sampling procedure.

We now upper bound $\EE[L_{k,r}]$. Note that $L_{k,r}$ denotes the number of bits generated by the following walk on line:
a $\widehat{\gamma}$-biased random walk on line starts at some point $x_r$, where $-T < x_r < T$ depends on $X_1, \dots, X_{\ell_r}$ and $T_{r,k}$, and stops after reaching either $-T$ or $T$.
The expected value of $L_{k,r}$ is upper bounded by the expected value of a $\widehat{\gamma}$-biased random walk on line starts at $x_r$ and stops after reaching either $-(T+ |x_r|)$ or $(T+|x_r|)$. Since $-T < x_r < T$, $\EE[L_{k,r}]$ is upper bounded by the expected value of a $\widehat{\gamma}$-biased random walk on line starts at $0$ and stops after reaching either $-2T$ or $2T$.

From Lemma~\ref{lemma: random walk with bias} we have the following bound on $\mu_{2T}$ in terms of $\mu_T$,
\begin{align*}
    \mu_{2T} \leq 12 \mu_T.
\end{align*}

The above expression, combined with Equation~\ref{eqn: main ub.}, implies that for every $k \in [n]$
\begin{align}
    \EE[X_1 + \dots + X_{\ell_r}] \leq T_k + 12 \mu_T, \label{eqn: result of main ub.}
\end{align}
where the expectation is over both $r$ (i.e. randomness used by $\widehat{A}$) and the randomness of the sampling procedure.

Next, we upper bound the expected number of times $B$ calls the sampling procedure, i.e. we upper bound $\EE[\ell_r]$. Let $I_t$ denote the random variable such that $I_t = 0$ if $t > \ell_r$ and $I_t = 1$ otherwise. Thus,
\begin{align*}
    \EE\left[\sum_{i = 1}^{\ell_r} X_i \right]
    &= \EE\left[\sum_{i = 1}^{\infty} X_i I_i \right] \\
    &= \sum_{i = 1}^{\infty} \EE[X_i I_i] \\
    &= \sum_{i = 1}^{\infty} \Pr[I_i = 1] \EE[X_i \mid I_i = 1] \\
    &= \sum_{i = 1}^{\infty} \Pr[\ell_r \geq i] \EE[X_i \mid \ell_r \geq i] \\
    &= \EE[X_i] \sum_{i = 1}^{\infty} \Pr[\ell_r \geq i] \\
    &= \EE[X_i] \EE[\ell_r].
\end{align*}
In the second last equation, we have used $\EE[X_i \mid \ell_r \geq i] = \EE[X_i]$. This is because the event $\ell_r \geq i$ depends on $X_1, \dots, X_{i-1}$ but not on $X_{i}$. Since $\EE[X_i] = \mu_T$ for all $i$, along with Equation~\ref{eqn: result of main ub.} we have
\begin{align*}
    \EE[\ell_r] \leq T_k/\mu_T + 12.
\end{align*}
Summing over all $k$ and using Equation~\ref{eqn: Tk sums to noisyR}, the expected number of queries made by algorithm $B$ is upper bounded by
\begin{align*}
    O\left( \frac{\noisyR(f)}{\widehat{\gamma}^2 \mu_T} + 12n \right).
\end{align*}
Thus the expected cost of $B$ is upper bounded by
\begin{align*}
    O\left( \frac{\noisyR(f)}{t\widehat{\gamma}^2 \mu_T} + \frac{12n}{t} \right).
\end{align*}
Since $\mu_T = \Omega(1/(t\widehat{\gamma}^2))$, the above quantity is upper bounded by
\begin{align*}
    O\left(\noisyR(f) + \frac{n}{t} \right).
\end{align*}

In order to complete the proof of the theorem, we show how to simulate $B$ to obtain a randomized query algorithm $B'$ for $f \circ \GapMaj_t$ of cost $O(t \cdot \noisyR(f) + n)$. The algorithm $B'$ simulates $B$ in the following manner: 
\begin{itemize}
    \item[1.] if $b$ queries $i$-th bit of $f$ with bias $1$, for some $i \in [n]$, then $B'$ queries all the bits of $i$-th Gap-Majority inner function.
    
    \item[2.] if $B$ queries $i$-th bit of $f$ with bias $1/\sqrt{t}$, for some $i \in [n]$, then $B'$ queries a random bit of $i$-th Gap-Majority inner function.
\end{itemize}
The correctness of $B'$ follows directly from the simulation and the correctness of $B$. Also, in both the cases (1. and 2.), cost paid by $B'$ is $t$ times the cost paid by $B$. Thus 
\begin{align}
    \R(f \circ \GapMaj_t) = O\left(t \cdot \noisyR(f) + n \right). \label{eqn: gapmaj ub bdb20}
\end{align}

\end{proof}

\begin{proof}[Proof of Lemma~\ref{lemma: random walk with bias}]
As mentioned in the statement of the lemma, the equality
    \begin{align}
    \mu_{T}
    &= \frac{T}{\widehat{\gamma}}\ -\ \frac{2T}{\widehat{\gamma}}
    (1 - \widehat{\gamma})^{T}
    \left(\frac{(1 + \widehat{\gamma})^{T} - (1 - \widehat{\gamma})^{T}}{(1 + \widehat{\gamma})^{2T} - (1 - \widehat{\gamma})^{2T}}\right), \label{eq: appendix, muT equality}
\end{align}
follows from~\cite[Chapter XIV]{feller}.

Next, we prove the second part of the lemma. The following inequalities are easy to observe.
\begin{align*}
    (1 + \widehat{\gamma})^{T} - (1 - \widehat{\gamma})^{T} &\leq \frac{2\widehat{\gamma}T}{1 - \widehat{\gamma}^2 T^2}, \\
    (1 + \widehat{\gamma})^{T} - (1 - \widehat{\gamma})^{T} &\geq 2 \widehat{\gamma} T , \\
    (1 - \widehat{\gamma})^{T} &\geq 1 - T\widehat{\gamma}. 
\end{align*}

Using the above inequalities and Equation~\ref{eq: appendix, muT equality} we get the following bounds on $\mu_T$
\begin{align}
    \frac{T^2}{(1 + \widehat{\gamma}T)} 
    \leq
    \mu_T \leq \frac{T}{\widehat{\gamma}}
    \left(
    \widehat{\gamma}T + 4 \widehat{\gamma}^2 T^2 - 4 \widehat{\gamma}^3 T^3
    \right). \label{eq: bounds on mu T appendix proof}
\end{align}

For $T = \Theta(1/\sqrt{t}\cdot\widehat{\gamma})$, the above expression implies that $\mu_T$ is lower bounded by $\Omega(1/(t \cdot \widehat{\gamma}^2))$.

Now we use Equation~\ref{eq: bounds on mu T appendix proof} to upper bound $\mu_{2T}$ in terms of $\mu_T$:
\begin{align*}
    \mu_{2T}
    &\leq \frac{2T}{\widehat{\gamma}}
    \left(
    2\widehat{\gamma}T + 16 \widehat{\gamma}^2 T^2 - 32 \widehat{\gamma}^3 T^3
    \right) \\
    &= \frac{T^2}{(1 + \widehat{\gamma}T)} \cdot \frac{(1 + \widehat{\gamma}T)}{T^2} \cdot \frac{2T}{\widehat{\gamma}}
    \left(
    2\widehat{\gamma}T + 16 \widehat{\gamma}^2 T^2 - 32 \widehat{\gamma}^3 T^3
    \right) \\
    &\leq \frac{(1 + \widehat{\gamma}T)}{T^2} \cdot \frac{2T}{\widehat{\gamma}}
    \left(
    2\widehat{\gamma}T + 16 \widehat{\gamma}^2 T^2 - 32 \widehat{\gamma}^3 T^3
    \right) \cdot \mu_T \\
    &= 2(1 + \widehat{\gamma}T) \cdot
    \left(
    2 + 16 \widehat{\gamma} T - 32 \widehat{\gamma}^2 T^2
    \right) \cdot \mu_T.
\end{align*}
Since $T = \Theta(1/(\sqrt{t}\cdot \widehat{\gamma}))$, for a suitable choice of constant, we have,
\begin{align*}
    \mu_{2T} \leq 12 \mu_T.
\end{align*}
\end{proof}

\end{document}